\documentclass[onecolumn, 12pt]{IEEETran}
\usepackage{color}
\usepackage{cite}
\usepackage{amsmath, amsfonts, amssymb}%, mnsymbol}
\usepackage{epsfig}
\usepackage{subfigure}
\usepackage{comment}
\usepackage{array}

\newcommand{\JHdagger}{*}
\newcommand{\JHnewpage}{}
\newcommand{\NN}{\nonumber}
\newcommand{\NNL}{\nonumber\\}

\newcommand{\AVR}[1]{\mathbb{E}\left[#1\right]}
\newcommand{\DoF}[1]{\lim_{P\to\infty}\frac{#1}{\log_2 P}}
\newcommand{\argmin}[1]{\underset{#1}{\arg\min}}
\newcommand{\argmax}[1]{\underset{#1}{\arg\max}}

\newcommand{\SNR}{\mathsf{\scriptscriptstyle SNR}}
\newcommand{\SINR}{\mathsf{\scriptscriptstyle SINR}}
\newcommand{\OIA}{\mathsf{\scriptscriptstyle OIA}}
\newcommand{\OIAI}{\mathsf{\scriptscriptstyle OIA1}}
\newcommand{\OIAII}{\mathsf{\scriptscriptstyle OIA2}}
\newcommand{\TDMII}{\mathsf{\scriptscriptstyle TDM2}}

\newcommand{\rate}{R}

\newtheorem{theorem}{Theorem}
\newtheorem{lemma}{Lemma}

\newtheorem{remark}{Remark}

\title{\huge On the Achievable DoF and User Scaling Law of Opportunistic
Interference Alignment in 3-Transmitter MIMO Interference Channels}

\author{\IEEEauthorblockN{Jung~Hoon~Lee},~\IEEEmembership{Student~Member,~IEEE},
\IEEEauthorblockN{Wan~Choi},~\IEEEmembership{Senior Member,~IEEE}
\thanks{J.~H.~Lee and W.~Choi are with Department of Electrical
    Engineering, Korea Advanced Institute of Science and Technology
    (KAIST), Daejeon 305-701, Korea (e-mail: tantheta@kaist.ac.kr,
    wchoi@ee.kaist.ac.kr).}
}

\begin{document}
\maketitle
\begin{abstract}
In this paper, we propose opportunistic interference alignment (OIA)
schemes for three-transmitter multiple-input multiple-output (MIMO)
interference channels (ICs). In the proposed OIA, each transmitter
has its own user group and selects a single user who has the most
aligned interference signals. The user dimensions provided by
multiple users are exploited to align interfering signals. Contrary
to conventional IA, perfect channel state information of all channel
links is not required at the transmitter, and each user just feeds
back one scalar value to indicate how well the interfering channels
are aligned. We prove that each transmitter can achieve the same
degrees of freedom (DoF) as the interference free case via user
selection in our system model that the number of receive antennas is
twice of the number of transmit antennas. Using the geometric
interpretation, we find the required user scaling to obtain an
arbitrary non-zero DoF. Two OIA schemes are proposed and compared
with various user selection schemes in terms of achievable rate/DoF
and complexity.

\end{abstract}

%\markboth{\today}{\today}

\begin{IEEEkeywords}
MIMO interference channel, interference alignment, opportunistic
interference alignment, postprocessing
\end{IEEEkeywords}

\clearpage

%%%%%%%%%%%%%%%%%%%%%%%%%%%%%%%%%%%%%%%%%%%%%%%%%%%%%%%%%%%%%%%%%%%
% % % % % % % % % % % % % % % % % % % % % % % % % % % % % % % % % %
%%%%%%%%%%%%%%%%%%%%%%%%%%%%%%%%%%%%%%%%%%%%%%%%%%%%%%%%%%%%%%%%%%%
\section{Introduction}

Interference alignment (IA) has been touted as a key technology for
handling interference in future wireless communications
\cite{ETW2008, CJ2008, JS2008, PH2011, J2012, KV2009, GCJ2011,
GCJ2008, ST2008, GJ2008}.
Contrary to the conventional schemes which orthogonalize
interference signals, \cite{CJ2008} showed that IA can achieve a
total of $N/2$ degrees of freedom (DoF) in an $N$-transmitter
single-input single-output (SISO) interference channel (IC).
The achievable DoF for $N$-transmitter MIMO has been found in
\cite{GJ2008}.
Despite the promising aspects of IA, its implementation has many
challenges.
IA generally requires the perfect global channel knowledge of
desired and interfering channels at the transmitter  which involves
excessive signal overheads although blind IA schemes \cite{J2012}
without requiring channel knowledge have recently been proposed for
some specific environments.
Imperfect channel state information significantly degrades the gain
of IA \cite{KV2009}. The large computation complexity necessitated
is also regarded as a big challenge for practical implementation.
The sub-optimality of IA in the practical operating SNR region is
another problem \cite{GCJ2011}.

Recently, IA techniques to ameliorate these problems have been
investigated. Iterative IA algorithms were proposed to optimize
precoding matrix and to reduce the global channel knowledge burden
based on channel reciprocity \cite{GCJ2008, PH2011}. To reduce
computational complexity, Suh and Tse \cite{ST2008} proposed a
subspace interference alignment technique for an uplink cellular
network system. In \cite{PFLD2010}, IA was opportunistically
performed in MIMO cognitive radio networks, where secondary
transmitters transmit their signals on only spatial dimensions not
used by primary transmitters.
IA with imperfect channel state information (CSI) was shown to
achieve the same DoF as IA with perfect CSI if the feedback size per
user is properly scaled \cite{TB2009, KV2009,AH2010}. Also, IA with
imperfect CSI in correlated channel was studied in \cite{MAH2011}.

%
%In \cite{TB2009}, a DoF of $N/2$ was shown to be achieved with
%limited feedback in an $N$-transmitter frequency-selective SISO IC
%if the feedback size per user is maintained as $N(L-1)\log P$ bits,
%where $L$ is the number of taps and $P$ is the total available
%transmit power.
%
%\cite{KV2009} showed that the maximum number of DoF is achievable
%with limited feedback in an $N$-user MIMO IC if the feedback size
%scales with the SNR.
%
%Also, full DoF is known to be achievable with analog feedback if the
%analog feedback power grows with the transmit power \cite{AH2010}.
%%
%IA with imperfect CSI in correlated channel was studied in
%\cite{MAH2011}.

%In \cite{CHHT2009}, interference aware-coordinated beamforming,
%considering interference channel to find beamforming and combining
%vectors, is proposed when two base stations support two users in a
%cellular system.
%
%In \cite{JPS2010}, the users which minimize generating minimum
%interference to other BSs, so called interference leakage, are
%selected among many users to achieve the optimal DoF in interfering
%multiple-access channel (IMAC).

Although there have been significant efforts to overcome the
practical challenges, the inherent shortcomings of IA highly
motivate the development of more practical techniques. It is
desirable to attain the promised gain of IA with limited feedback
and reduced computational complexity. In this context, interference
management by user selection attracts attentions. The key idea
behind this opportunistic interference management is to select and
serve the user with the best channel or interference condition.
The selection criteria include maximum signal-to-noise ratio (SNR),
minimum interference-to-noise ratio (INR), maximum
signal-to-interference-pulse-noise ratio (SINR), and so on
\cite{SH2005, YG2006, CA2008, RRCY2010}.

%%%%%%%%%%%%%%%%%%%%%%%%%%%%%%%%%%%%%%%%%%%%%%%%%%%%%%%%%%%%%%%%%%%
% % % % % % % % % % % % % % % % % % % % % % % % % % % % % % % % % %
%%%%%%%%%%%%%%%%%%%%%%%%%%%%%%%%%%%%%%%%%%%%%%%%%%%%%%%%%%%%%%%%%%%
\subsection{Opportunistic Interference Alignment (OIA)}

In this paper, we propose opportunistic interference alignment (OIA)
schemes by interpreting the opportunistic interference management
from a perspective of IA. In our proposed OIA, the user dimensions
provided by multiple users is exploited to align interfering
signals.  Different forms of OIA have been proposed in $K$-user SISO
IC using the random phase offset \cite{NELL2010} and in a cognitive
radio network \cite{SF2011}.

There are three transmitters and three user groups associated with
the transmitters.
Each user feeds one scalar value of an \emph{interference alignment
measure} back to the own transmitter, which indicates how well the
interfering channels are well aligned.
Based on the feedback information, each transmitter selects and
serves only a single user whose interfering channels are most
aligned so that a three-transmitter MIMO IC is opportunistically
constructed.
Thus, interference alignment is achieved by user selection rather
than transmit beamforming. Collaboration and Information sharing
among transmitters are not required.

The proposed OIA combines the concepts of opportunistic beamforming
and IA.
Contrary to opportunistic beamforming in a MIMO broadcast channel
\cite{HS2007, KGS2008}, each user only considers the interfering
channels rather than the desired channel; the interference from one
transmitter helps the other transmitter's user selection.

The basic concept of OIA was roughly introduced in 3-transmitter
$2\times 2$ MIMO IC \cite{LC2010} and $M\times 2M$ MIMO IC
\cite{LC2011}. However, the maximum achievable DoF by the OIA and
the relationship between the achievable DoF and the required user
scaling were not found. In this paper, we generalize our preliminary
studies on OIA \cite{LC2010, LC2011}. We consider the
three-transmitter $N_T \times N_R$ MIMO IC where $(N_T, N_R) = (M,
2M)$ and show that each transmitter can obtain DoF up to $M$ via the
proposed user selection. We also derive the required user scaling to
obtain given DoF.

For implementation, we propose two OIA schemes. In the first OIA
scheme (OIA1), each user directly minimizes the rate loss induced by
the interfering channels. Thus, each transmitter selects a user with
the minimum rate loss.
In the second OIA scheme (OIA2), aligned level of interfering
channels is geometrically interpreted; the transmitter selects a
user whose interfering channels span the closest subspaces.
The complexity of OIA2 can be reduced compared to OIA1 through a
geometric interpretation.

%%%%%%%%%%%%%%%%%%%%%%%%%%%%%%%%%%%%%%%%%%%%%%%%%%%%%%%%%%%%%%%%%%%
% % % % % % % % % % % % % % % % % % % % % % % % % % % % % % % % % %
%%%%%%%%%%%%%%%%%%%%%%%%%%%%%%%%%%%%%%%%%%%%%%%%%%%%%%%%%%%%%%%%%%%
\subsection{Contributions}
We investigate the achievable DoF and user scaling law of the OIA
scheme in a three-transmitter MIMO IC where $K$ users are associated
with each transmitter, and the selected users together with their
transmitters construct a three-transmitter MIMO IC.
In our system model, each transmitter sends $M$ streams with $N_T
(=M)$ antennas, and each receiver has $N_R(=2M)$ antennas.
\begin{itemize}
\item We prove that each transmitter can achieve DoF $M$
by the OIA schemes without symbol extension and no cooperation. In
this case, we show that the transmitter and the selected user act
like interference-free $M\times M$ point-to-point MIMO system.
For $M \times 2M$ MIMO IC composed of three transmitters and three
users, $2M/3$ DoF per user is known to be achievable (with perfect
CSIT and symbol extension) \cite{GJ2008}.
Our result seems to be contradictory at first glance, but the
required spatial dimensions for $M$ data streams are secured through
the user dimensions provided by the $K$ users. This means that
multiuser DoF are translated into IA spatial dimensions.

\item We show that the number of users associated with each transmitter,
$K$, is enough to be scaled as $K \propto P^{mM}$ to achieve
$m\in[0,M]$ DoF per transmitter. When $K$ is fixed, the achievable
DoF by the OIA schemes is proved to be zero.

\item Finally, we look into the practical advantages of the proposed
OIA schemes; we show that the OIA scheme based on geometric concept
significantly reduces the computational complexity while achieving a
notable rate improvement compared to conventional opportunistic user
selection schemes.

\end{itemize}
%

%Since each transmitter sends $M$ streams, each user needs $3M$
%antennas to obtain a DoF of $M$.
%
%On the other hand, if we can align two $M$-dimensional interfering
%signals into an $M$-dimensional subspace, $2M$ receive antennas are
%enough to obtain a DoF of $M$ which is our system configuration.
%
%To align the interfering channels, the conventional IA schemes not
%considering the user dimension requires perfect CSIT, large
%computational complexity, and the symbol extension. Moreover, in our
%antenna configuration, the most aligned interfering channels spans
%$4M/3$-dimensional subspace, so each user can obtain a DoF of
%$2M/3$.

%In our proposed scheme, however, we exploit the user dimension to
%align the interferences. Each user needs to find the interference
%alignment measure which is a scalar value used for user selection.

%%%%%%%%%%%%%%%%%%%%%%%%%%%%%%%%%%%%%%%%%%%%%%%%%%%%%%%%%%%%%%%%%%%
% % % % % % % % % % % % % % % % % % % % % % % % % % % % % % % % % %
%%%%%%%%%%%%%%%%%%%%%%%%%%%%%%%%%%%%%%%%%%%%%%%%%%%%%%%%%%%%%%%%%%%
\subsection{Organization}

The rest of this paper is organized as follows. Our system model is
described in Section II.
Preliminaries about the angles between two subspaces are provided in
Section III.
The proposed OIA schemes are described in Section IV, and the
achievable rate and DoF are analyzed in Section V.
Several conventional opportunistic user selection schemes are
summarized and compared with OIA schemes in Section VI.
The conclusions and comments on areas of future interest are given
in Section VII.

%%%%%%%%%%%%%%%%%%%%%%%%%%%%%%%%%%%%%%%%%%%%%%%%%%%%%%%%%%%%%%%%%%%
% % % % % % % % % % % % % % % % % % % % % % % % % % % % % % % % % %
%%%%%%%%%%%%%%%%%%%%%%%%%%%%%%%%%%%%%%%%%%%%%%%%%%%%%%%%%%%%%%%%%%%
\subsection{Notations}

Throughout the paper, the notations $\mathbf{A}^\JHdagger$,
$\lambda_n(\mathbf{A})$, $\mathbf{v}_n(\mathbf{A})$,
$tr(\mathbf{A})$ and $\Vert \mathbf{A} \Vert_F$ denote the conjugate
transpose, $n$th largest eigenvalue, eigenvector corresponding to
$\lambda_n(\mathbf{A})$, trace, and Frobenius norm of matrix
$\mathbf{A}$, respectively.
Also, the notations $\mathbf{I}_n$, $diag(\cdot)$, $\mathbb{C}^{n}$
and $\mathbb{C}^{m\times n}$ indicate the $n\times n$ identity
matrix, a diagonal matrix whose diagonal elements are $(\cdot)$, the
$n$-dimensional complex space, and the set of $m\times n$ complex
matrices, respectively.
%
%For two functions $f(x)$ and $g(x)$, we use big-O notation
%$f(x)=\bigO{g(x)}$ if and only if there exists a positive real
%number $c$ such that $\vert f(x)\vert \le c\vert g(x) \vert$.

%%%%%%%%%%%%%%%%%%%%%%%%%%%%%%%%%%%%%%%%%%%%%%%%%%%%%%%%%%%%%%%%%%%
% % % % % % % % % % % % % % % % % % % % % % % % % % % % % % % % % %
%%%%%%%%%%%%%%%%%%%%%%%%%%%%%%%%%%%%%%%%%%%%%%%%%%%%%%%%%%%%%%%%%%%
\JHnewpage\section{System Model}

Our system model is depicted in Fig. \ref{fig:system_model}. There
are three transmitters having $N_T (= M)$ antennas, and each
transmitter has its own user group consisting of $K$ users with $N_R
(=2M)$ antennas each.
Each transmitter selects a single user in its own user group and
sends $M$ data streams to the selected user. Consequently, the
transmitters and their selected users construct a three-transmitter
MIMO IC.
For user selection, each transmitter uses only partial information
fed back from each user, which is a single scalar value.
No collaborations and no information sharing are allowed among the
transmitters.

Our system operates with following four steps:
\begin{itemize}
\item Step 1: Each transmitter broadcasts a reference signal.
\item Step 2: Each user feeds one analog value back to the own transmitter.
\item Step 3: Each transmitter selects one user in its user group.
\item Step 4: Each transmitter serves the selected user with the random beams.
\end{itemize}
In Step 1, each transmitter broadcasts a reference signal. Thus,
each user obtains the information of the desired channel and two
interfering channels.
In Step 2, each user generates the feedback information from the
channel information, which is one scalar value.
Various feedback information can be constructed according to the
postprocessing and the user selection schemes.
In Step 3, each transmitter selects a single user in its user group.
Note that the user selection at each transmitter is independent of
one another because there are no information sharing and
collaboration among the transmitters.
In Step 4, the transmitters serve the selected users with the random
beams. Thus, the three-transmitter MIMO IC is opportunistically
constructed.

Since a user selection at each transmitter does not affect the
performances of the other transmitters, without loss of generality,
we only consider the user selection at the first transmitter.
Other transmitters can achieve the same average achievable rate with
the identical setting.

At the $k$th user in the first user group, the received signal
denoted by $\mathbf{y}_k$ is given by
\begin{align}
 \mathbf{y}_k
%    &=\mathbf{H}'_{k,1}\mathbf{W}_1 \mathbf{x}_1
%        + \sum_{i=2}^{3}\mathbf{H}'_{k,i}\mathbf{W}_i\mathbf{x}_i
%        + \mathbf{n}_k \NNL
    &=\mathbf{H}_{k,1} \mathbf{x}_1
        + \sum_{i=2}^{3} \mathbf{H}_{k,i}\mathbf{x}_i
        + \mathbf{n}_k, \label{eqn:y_k}
\end{align}
where $\mathbf{H}_{k,i} \in \mathbb{C}^{N_R\times M}$ is the channel
matrix from transmitter $i$ to user $k$ in the first user group. The
term $\mathbf{x}_i \in \mathbb{C}^{M\times 1}$ is the transmit
signal of the $i$th transmitter. Since each transmitter does not
have channel state information, we assume equal power allocation
among $M$ data streams, i.e., $\mathbb{E}\{ \mathbf{x}_i
\mathbf{x}_i ^\JHdagger\} = (P/M) \mathbf{I}_M$. The random vector
$\mathbf{n}_k \in \mathbb{C} ^{N_R\times 1}$ is Gaussian noise with
zero mean and an identity covariance matrix, i.e., $\mathbf{n}_k
\sim \mathcal{CN} (0,\mathbf{I}_{N_R})$. When $N_T>M$, the system
model becomes statistically identical with $M\times 2M$ MIMO IC if
each transmitter uses an arbitrary precoding matrix $\mathbf{W} \in
\mathbb{C} ^{N_T\times M}$ such that $\mathbf{W}^\JHdagger
\mathbf{W}=\mathbf{I}_M$.

From \eqref{eqn:y_k}, the capacity at the $k$th user is given by
\cite{SPB2008}
\begin{align}
 C_k &=\log_2
    \bigg\vert
    \mathbf{I}_{N_R}
    +\frac{P}{M}
    \mathbf{H}_{k,1} \mathbf{H}^\JHdagger_{k,1}
    \bigg(
        \mathbf{I}_{N_R} + \frac{P}{M}\sum_{i=2}^3
        \mathbf{H}_{k,i} \mathbf{H}^\JHdagger_{k,i}
    \bigg)^{-1}
    \bigg\vert,
    \label{eqn:optimal_capacity_ik}
\end{align}
which requires joint decoding and non-linear receivers.
In our system model, we assume that each user adopts linear
postprocessing. Half of receive antenna dimensions (i.e., $M$) are
used for the desired $M$ data streams, and the remaining dimensions
are used for interference suppression. The received signals are
projected onto the $M$-dimensional subspace designated for the
desired signals at each receiver. The $k$th user uses the
postprocessing matrix $\mathbf{F}_k \in \mathbb{C}^{M\times N_R}$ to
project the received signals onto the row space of $\mathbf{F}_k$
which is $M$-dimensional subspace designated for the desired signals
in $\mathbb{C}^{N_R}$. Therefore, $\mathbf{F}_k$ consists of the
bases of the $M$-dimensional subspace designated for the desired
signals and satisfies $\mathbf{F}_k \mathbf{F}_k^\JHdagger =
\mathbf{I}_M$.  In this way, when each transmitter selects the user
who has perfectly aligned interfering signals, the selected user can
obtain DoF $M$ by the postprocessing.

At the $k$th user, the received signal after postprocessing becomes
\begin{align}
\mathbf{F}_k\mathbf{y}_k
    =\mathbf{F}_k\mathbf{H}_{k,1}\mathbf{x}_1
    +\sum_{i=2}^3 \mathbf{F}_k\mathbf{H}_{k,i}\mathbf{x}_i
     + \mathbf{F}_k\mathbf{n}_k\NN
\end{align}
and the achievable rate at the user $k$ denoted by $\rate_k$ is
given by
\begin{align}
 \rate_k
 &=\log_2
    \bigg\vert
    \mathbf{I}_M
    +\frac{P}{M}\mathbf{F}_k
    \mathbf{H}_{k,1} \mathbf{H}^\JHdagger_{k,1}
    \mathbf{F}_k^\JHdagger\bigg(
        \mathbf{I}_M +
        \frac{P}{M}\sum_{i=2}^3 \mathbf{F}_k\mathbf{H}_{k,i}
        \mathbf{H}^\JHdagger_{k,i} \mathbf{F}_k^\JHdagger
    \bigg)^{-1}
    \bigg\vert\NNL
 &= \log_2\frac{
    \left\vert
    \mathbf{I}_M
    +\frac{P}{M}\sum_{i=1}^3 \mathbf{F}_k\mathbf{H}_{k,i}
        \mathbf{H}^\JHdagger_{k,i}\mathbf{F}_k^\JHdagger
    \right\vert}
    {\left\vert
        \mathbf{I}_M +
        \frac{P}{M}\sum_{i=2}^3 \mathbf{F}_k\mathbf{H}_{k,i}
        \mathbf{H}^\JHdagger_{k,i}\mathbf{F}_k^\JHdagger
    \right\vert}. \label{eqn:capacity_ik}
\end{align}

%%%%%%%%%%%%%%%%%%%%%%%%%%%%%%%%%%%%%%%%%%%%%%%%%%%%%%%%%%%%%%%%%%%
% % % % % % % % % % % % % % % % % % % % % % % % % % % % % % % % % %
%%%%%%%%%%%%%%%%%%%%%%%%%%%%%%%%%%%%%%%%%%%%%%%%%%%%%%%%%%%%%%%%%%%
Let $k^\star$ be the index of the selected user at the first
transmitter. Then, the achievable rate of the first transmitter
becomes $\rate_{k^\star}$.
When the transmitter supports one of $K$ users, the average
achievable rate at the first transmitter denoted by
$\mathcal{R}_{[K]}$ becomes
\begin{align}
 \mathcal{R}_{[K]} &\triangleq \mathbb{E}_{\mathbf{H}}[\rate_{k^\star}],
\end{align}
In this case, the achievable DoF of the first transmitter becomes
\begin{align}
% \mathcal{D} \triangleq \DoF{\AVR{\rate_{k^\star}}}.
 \mathcal{D} \triangleq \DoF{\mathcal{R}_{[K]}}.
\end{align}
Note that the average achievable rate and DoF of the system with all
transmitters become $3\mathcal{R}_{[K]}$ and $3\mathcal{D}$,
respectively.

Throughout the paper, we assume that all channel matrices (i.e.,
$\mathbf{H}_{k,i}$ for all $k$ and $i$) have independent and
identically distributed (i.i.d.) elements so that the interfering
subspaces formed by the interfering channels are isotropic and
independent of each other.

%%%%%%%%%%%%%%%%%%%%%%%%%%%%%%%%%%%%%%%%%%%%%%%%%%%%%%%%%%%%%%%%%%%
% % % % % % % % % % % % % % % % % % % % % % % % % % % % % % % % % %
%%%%%%%%%%%%%%%%%%%%%%%%%%%%%%%%%%%%%%%%%%%%%%%%%%%%%%%%%%%%%%%%%%%
\JHnewpage\section{Preliminaries -- Angles between Two Subspaces}

In our system, each user suffers from two interfering channels each
of which constructs an $M$-dimensional subspace in $\mathbb{C}
^{N_R}$.
Because the distance between the two subspaces can be measured in
terms of angles between them, we shortly overview the angles between
two subspaces.
As a widely used geometric concept in wireless communications, the
Grassmann manifold $\mathcal{G}_{N_R,M}(\mathbb{C})$ is defined as
the set of all $M$-dimensional subspaces in an $N_R$-dimensional
space, $\mathbb{C}^{N_R}$ \cite{ZT2002, LHS2003, DLR2008, DLLR2009,
RJ2008}.
Consider two $M$-dimensional subspaces $\mathcal{A}, \mathcal{B}$ in
$N_R$-dimensional space, i.e., $\mathcal{A}, \mathcal{B} \in
\mathcal{G}_{N_R,M}(\mathbb{C})$.
The angles between the subspaces can be measured with the
\emph{principal angles} that is also called as the canonical angles.
Since both $\mathcal{A}$ and $\mathcal{B}$ are $M$-dimensional
subspaces, there are $M$ principal angles between them. Let
$\theta_1, \ldots, \theta_M \in[0,\pi/2]$ be the $M$ principal
angles such that $\theta_1< \ldots< \theta_M$, then we can find them
recursively by searching $N_R$-dimensional unit vectors
$\{\mathbf{a}_m, \mathbf{b}_m\}_{m=1}^M$ such that \cite[Chap.
12]{GL1989}
\begin{align}
 \cos\theta_m
 =\max_{\mathbf{a}\in\mathcal{A} \atop \mathbf{b}\in\mathcal{B}}~
    \vert\mathbf{a}^\JHdagger\mathbf{b}\vert
 =\vert\mathbf{a}_m^\JHdagger\mathbf{b}_m\vert \NN
\end{align}
subject to $\Vert\mathbf{a}\Vert=1$, $\Vert\mathbf{b}\Vert=1$,
$\mathbf{a}^\JHdagger\mathbf{a}_n=0$, $\mathbf{b} ^\JHdagger
\mathbf{b}_n=0$ ($1\le n \le m-1$).
The vectors $\{\mathbf{a}_m\}_{m=1}^M$ and
$\{\mathbf{b}_m\}_{m=1}^M$ become the \emph{principal vectors} of
$\mathcal{A}$ and $\mathcal{B}$, respectively.

From the principal angles, we can define various distances between
the subspaces.
Arguably, the \emph{chordal distance} is the most widely used one
among them.
The chordal distance between the subspaces $\mathcal{A}$ and
$\mathcal{B}$ denoted by $d_c(\mathcal{A}, \mathcal{B})$ is defined
as
\begin{align}
 d_c(\mathcal{A},\mathcal{B})
 \triangleq
%    \left(\sum_{m=1}^M \sin^2\theta_m\right)^{1/2}.
    \sqrt{\sum_{m=1}^M \sin^2\theta_m}.
    \label{eqn:chordal_distance1}
\end{align}

Alternatively, we can use the \emph{generator matrices} to represent
the chordal distance; a generator matrix of a subspace consists of
orthonormal columns that span the subspace.
For example, $\mathbf{A}, \mathbf{B}\in\mathbb{C}^{N_R\times M}$ are
generator matrices of the subspace $\mathcal{A}, \mathcal{B} \in
\mathcal{G}_{N_R, M}(\mathbb{C})$ when $\mathbf{A}^\JHdagger
\mathbf{A}= \mathbf{B}^\JHdagger \mathbf{B} =\mathbf{I}_M$, and
their columns span the subspaces $\mathcal{A}$ and $\mathcal{B}$,
respectively.
Although the generator matrices $\mathbf{A}$ and $\mathbf{B}$ are
infinitely many, the chordal distance between two subspaces is
uniquely obtained with any generator matrix pairs such that
\begin{align}
    d_c(\mathcal{A},\mathcal{B})
    &=\frac{1}{2}\Vert \mathbf{A}\mathbf{A}^\JHdagger
        - \mathbf{B}\mathbf{B}^\JHdagger \Vert_F \NNL
    &=\sqrt{M-tr(\mathbf{A}^\JHdagger \mathbf{B} \mathbf{B}^\JHdagger
    \mathbf{A})}.\label{eqn:chordal_definition}
\end{align}

Also, we can obtain the principal angles and the principal vectors
from the generator matrices.
Let the singular value decomposition (SVD) of
$\mathbf{A}^\JHdagger\mathbf{B}$ be \cite[Chap. 12]{GL1989}
\begin{align}
    \mathbf{A}^\JHdagger\mathbf{B}=\mathbf{YDZ}^\JHdagger,
    \label{eqn:YZ}
\end{align}
where $\mathbf{Y},\mathbf{Z}\in\mathbb{C}^{M\times M}$ are unitary
matrices and $\mathbf{D}=diag(\mu_1,\mu_2,\ldots,\mu_M)$ where
$\mu_m$ is the $m$th largest singular value such that $\mu_1 \ge
\mu_2 \ldots \ge \mu_M \ge 0$. Then, the $m$th largest singular
value of $\mathbf{A}^\JHdagger\mathbf{B}$ and the $m$th principal
angle between $\mathcal{A}$ and $\mathcal{B}$ has the following
relationship:
\begin{align}
    \mu_m = \cos\theta_m. \NN
\end{align}
Also, the corresponding principal vectors $\mathbf{a}_m$ and
$\mathbf{b}_m$ can be obtained from $\mathbf{Y}$ and $\mathbf{Z}$
such that
\begin{align}
 \mathbf{a}_m=\mathbf{A}\mathbf{y}_m,\quad
 \mathbf{b}_m=\mathbf{B}\mathbf{z}_m, \NN
\end{align}
where $\mathbf{y}_m$ and $\mathbf{z}_m$ are the $m$th column vectors
of $\mathbf{Y}$ and $\mathbf{Z}$, respectively.

%In our system model, each user has two interfering channels, and
%each interfering channel spans an $M$-dimensional subspace in
%$\mathbb{C}^{N_R}$ with probability 1 because the elements of each
%$N_R$-dimensional interfering channel vector are i.i.d. and
%circularly symmetric complex Gaussian random variables.
%%
%Thus, the union of the two subspaces formed by the interfering
%channels is a $2M(=N_R)$-dimensional space with probability 1 for
%the same reason.
%%
%Using the same notations, let $\mathcal{A}$ and $\mathcal{B}$ be the
%subspaces formed by the interfering channels at a user and
%$\mathbf{A}$ and $\mathbf{B}$ be the generator matrices of them.
%

From the generator matrices and the principal angles, we obtain the
following lemma needed to analyze the proposed OIA scheme.

%%%%%%%%%%%%%%%%%%%%%%%%%%%%%%%%%%%%%%%%%%%%%%%%%%%%%%%%%%%%%%%%%%%
\begin{lemma}\label{lemma:p_angles}
When $\mathbf{A}, \mathbf{B} \in\mathbb{C}^{N_R\times M}$ are the
generator matrices of the subspaces $\mathcal{A}, \mathcal{B}
\in\mathcal{G}_{N_R, M}(\mathbb{C})$,
the eigenvalues of $\mathbf{A} \mathbf{A}^\JHdagger + \mathbf{B}
\mathbf{B} ^\JHdagger$ can be represented in descending order as
\begin{align}
    \underbrace{1+\cos^2\theta_1,\ldots,1+\cos^2\theta_M}_{M},
    \underbrace{1-\cos^2\theta_M,\ldots,1-\cos^2\theta_1}_{M}
    %\underbrace{0,\ldots,0}_{N_R-2K}.
    \label{eqn:AA_BB_eigen}
\end{align}
where $\theta_m$ is the $m$th principal angle between $\mathcal{A}$
and $\mathcal{B}$.
\end{lemma}
\vspace{.1in}
\begin{proof}%See Appendix \ref{appendix:p_angles}.%
Using the unitary matrices $\mathbf{Y}$ and $\mathbf{Z}$ in
\eqref{eqn:YZ}, $\mathbf{A} \mathbf{A}^\JHdagger + \mathbf{B}
\mathbf{B} ^\JHdagger$ can be rewritten as
\begin{align}
    \mathbf{A} \mathbf{A}^\JHdagger +\mathbf{B} \mathbf{B}^\JHdagger
    &=\mathbf{AY} (\mathbf{AY})^\JHdagger +  \mathbf{BZ}(\mathbf{BZ})^\JHdagger\NNL
    &=\sum_{m=1}^{M}\left(\mathbf{a}_m\mathbf{a}_m^\JHdagger
    +\mathbf{b}_m\mathbf{b}_m^\JHdagger\right).
    \label{eqn:AA_BB}
\end{align}
Also, we decompose $\mathbf{b}_m$ as
\begin{align}
    \mathbf{b}_m= \cos\theta_m\mathbf{a}_m+\sin\theta_m\mathbf{e}_m,
    \label{eqn:bb}
\end{align}
where $\theta_m$ is the $m$th principal angle, and $\mathbf{e}_m$ is
an unit vector orthogonal with $\mathbf{a}_m$ such that
$\Vert\mathbf{e}_m\Vert=1$ and $\mathbf{a}_m \perp \mathbf{e}_m$.

From the property of principal vectors, $\mathbf{a}_i \perp
\mathbf{a}_j$ and $\mathbf{b}_i \perp \mathbf{b}_j$ for $i\ne j$.
Also, from the relationship between the principal angle and the
principal vector given in \eqref{eqn:YZ}, it is satisfied that
\begin{align}
(\mathbf{AY})^\JHdagger\mathbf{BZ}
=[\mathbf{a}_1,\ldots,\mathbf{a}_{M}]^\JHdagger
[\mathbf{b}_1,\ldots,\mathbf{b}_{M}] =\mathbf{D}\NN
\end{align}
which implies $\mathbf{a}_i \perp \mathbf{b}_j$ for $i\ne j$ because
$\mathbf{D}$ is a diagonal matrix defined in \eqref{eqn:YZ}.

Since $\mathbf{a}_i \perp \{\mathbf{a}_j, \mathbf{b}_j\}$ and
$\mathbf{b}_i \perp \{\mathbf{a}_j, \mathbf{b}_j\}$ for $i\ne j$, it
is satisfied that $span(\mathbf{a}_i, \mathbf{b}_i) \perp
span(\mathbf{a}_j, \mathbf{b}_j)$ for $i\ne j$, equivalently,
$span(\mathbf{a}_i, \mathbf{e}_i) \perp span(\mathbf{a}_j,
\mathbf{e}_j)$ for $i\ne j$.
Also, from the fact that $\mathbf{a}_i \perp \mathbf{e}_i$, we can
conclude that $\{\mathbf{a}_1, \ldots, \mathbf{a}_M, \mathbf{e}_1,
\ldots, \mathbf{e}_M \}$ becomes $2M$ orthonormal bases of
$\mathbb{C}^{2M}$.

From \eqref{eqn:bb}, we have
\begin{align}
    \mathbf{b}_m\mathbf{b}_m^\JHdagger
    &=(\cos\theta_m\mathbf{a}_m+\sin\theta_m\mathbf{e}_m)
        (\cos\theta_m\mathbf{a}_m+\sin\theta_m\mathbf{e}_m)^\JHdagger\NNL
    &=\cos^2\theta_m\cdot\mathbf{a}_m\mathbf{a}_m^\JHdagger
        +\sin^2\theta_m\cdot\mathbf{e}_m\mathbf{e}_m^\JHdagger,\NN
%    &=\mu_i^2\mathbf{a}_i\mathbf{a}_i^\JHdagger
%        + (1-\mu_i^2) \mathbf{e}_i\mathbf{e}_i^\JHdagger,\NN
\end{align}
and \eqref{eqn:AA_BB} can be rewritten by
\begin{align}
    \mathbf{A} \mathbf{A}^\JHdagger +\mathbf{B} \mathbf{B}^\JHdagger
    &=\sum_{m=1}^{M}\left(\mathbf{a}_m\mathbf{a}_m^\JHdagger
        +\mathbf{b}_m\mathbf{b}_m^\JHdagger\right)\NNL
    &=\sum_{m=1}^{M}\left[(1+\cos^2\theta_m)\mathbf{a}_m\mathbf{a}_m^\JHdagger
        +(1-\cos^2\theta_m)\mathbf{e}_m\mathbf{e}_m^\JHdagger\right].\NN
\end{align}
Thus, $\mathbf{A} \mathbf{A}^\JHdagger +\mathbf{B}
\mathbf{B}^\JHdagger$ has the eigenvectors
$\{\mathbf{a}_m\}_{m=1}^M$ and $\{\mathbf{e}_m\}_{m=1}^M$, and
ordered eigenvalues given in \eqref{eqn:AA_BB_eigen}.
\end{proof}
%%%%%%%%%%%%%%%%%%%%%%%%%%%%%%%%%%%%%%%%%%%%%%%%%%%%%%%%%%%%%%%%%%%
\vspace{.1in}
%%%%%%%%%%%%%%%%%%%%%%%%%%%%%%%%%%%%%%%%%%%%%%%%%%%%%%%%%%%%%%%%%%%
% % % % % % % % % % % % % % % % % % % % % % % % % % % % % % % % % %
%%%%%%%%%%%%%%%%%%%%%%%%%%%%%%%%%%%%%%%%%%%%%%%%%%%%%%%%%%%%%%%%%%%
\begin{lemma}\label{lemma:p_angles_sum}
When $\mathbf{A}, \mathbf{B} \in\mathbb{C}^{N_R\times M}$ are the
generator matrices of the subspaces $\mathcal{A}, \mathcal{B}
\in\mathcal{G}_{N_R, M}(\mathbb{C})$, sum of the $M$ smallest
eigenvalues of $\mathbf{A} \mathbf{A}^\JHdagger + \mathbf{B}
\mathbf{B} ^\JHdagger$ becomes the squared chordal distance between
$\mathcal{A}$ and $\mathcal{B}$.
\end{lemma}
%%%%%%%%%%%%%%%%%%%%%%%%%%%%%%%%%%%%%%%%%%%%%%%%%%%%%%%%%%%%%%%%%%%
% % % % % % % % % % % % % % % % % % % % % % % % % % % % % % % % % %
%%%%%%%%%%%%%%%%%%%%%%%%%%%%%%%%%%%%%%%%%%%%%%%%%%%%%%%%%%%%%%%%%%%
\begin{proof}
From Lemma \ref{lemma:p_angles}, we can find that
\begin{align}
 \sum_{m=M+1}^{2M} \lambda_m \left(\mathbf{A}\mathbf{A}^\JHdagger
   + \mathbf{B} \mathbf{B} ^\JHdagger\right)
 &=\sum_{m=1}^{M} (1 - \cos^2\theta_m) \NNL
 &\stackrel{(a)}{=} d_c^2\left(\mathcal{A}, \mathcal{B}\right)\NN,
\end{align}
where the equality $(a)$ is from the definition of the chordal
distance given in \eqref{eqn:chordal_distance1}.
\end{proof}

%%%%%%%%%%%%%%%%%%%%%%%%%%%%%%%%%%%%%%%%%%%%%%%%%%%%%%%%%%%%%%%%%%%
% % % % % % % % % % % % % % % % % % % % % % % % % % % % % % % % % %
%%%%%%%%%%%%%%%%%%%%%%%%%%%%%%%%%%%%%%%%%%%%%%%%%%%%%%%%%%%%%%%%%%%

%%%%%%%%%%%%%%%%%%%%%%%%%%%%%%%%%%%%%%%%%%%%%%%%%%%%%%%%%%%%%%%%%%%
% % % % % % % % % % % % % % % % % % % % % % % % % % % % % % % % % %
%%%%%%%%%%%%%%%%%%%%%%%%%%%%%%%%%%%%%%%%%%%%%%%%%%%%%%%%%%%%%%%%%%%
\JHnewpage\section{Opportunistic Interference Alignment}
\subsection{What is the Opportunistic Interference Alignment?}

The basic concept of interference alignment is to minimize the
dimensions occupied by interfering signals. Although the dimensions
of \emph{each interfering signal} are irreducible, the dimensions
occupied by \emph{all interfering signals} can be minimized by
aligning them into the same subspace.
When the number of users is finite, it is obvious that two
interfering channels at each user are not aligned because two
interfering transmitters cannot access a common subspace at each
receiver. However, as the number of users increases, we can find the
user whose interfering channels are more overlapped with each other.
In the proposed OIA, we exploit the multiuser dimensions to align
the interfering signals.
By opportunistic user selection, two irreducible $M$-dimensional
interfering signals can be aligned in an $M$-dimensional subspace.

In this section, we propose two different OIA schemes. In the first
OIA scheme, the transmitter selects a user whose rate loss caused by
interference is minimum. In the second OIA scheme, the transmitter
selects a user who has the minimum distance between the interfering
signals. Now, we assume that the elements of all channel matrices
are i.i.d. circularly symmetric complex Gaussian random variables
with zero mean and unit variance.

We decompose the achievable rate at each user given in
\eqref{eqn:capacity_ik} into two terms $\rate_k^+$ and $\rate_k^-$
given by
\begin{align}
\rate_k^{+}
    &=\log_2
    \bigg\vert
    \mathbf{I}_M
        +\frac{P}{M}\sum_{i=1}^{3}
        \mathbf{F}_k\mathbf{H}_{k,i} \mathbf{H}_{k,i}^\JHdagger
    \mathbf{F}_k^\JHdagger
    \bigg\vert \label{eqn:capacity_gain}\\
\rate_k^{-}
    &=\log_2\bigg\vert
        \mathbf{I}_M +
        \frac{P}{M}\sum_{i=2}^{3} \mathbf{F}_k\mathbf{H}_{k,i}
        \mathbf{H}_{k,i}^\JHdagger\mathbf{F}_k^\JHdagger
    \bigg\vert, \label{eqn:capacity_loss}
\end{align}
respectively, so that $\rate_k = \rate_k^{+} - \rate_k^{-}$. We call
$\rate_k^-$ as \emph{rate loss term}.
In the same way, we can rewrite the average achievable rate at the
selected user among $K$ users as
\begin{align}
 \mathcal{R}_{[K]} = \mathcal{R}_{[K]}^+ - \mathcal{R}_{[K]}^-,
\end{align}
where $\mathcal{R}_{[K]}^+ = \AVR{\rate_{k^\star}^+}$ and
$\mathcal{R}_{[K]}^- = \AVR{\rate_{k^\star}^-}$, respectively.

Our proposed OIA schemes aim at minimizing the dimension occupied by
the interfering signals and hence maximizing the achievable DoF at
the transmitter.
Since it is straightforward that $\lim_{P\to\infty}
(\mathcal{R}_{[K]}^+/\log_2P)=M$, the achievable DoF of the first
transmitter using OIA can be expressed by
\begin{align}
 \DoF{\mathcal{R}_{[K]}}
 %&= \DoF{\mathcal{R}_{[K]}^+} - \DoF{\mathcal{R}_{[K]}^-}\NNL
 &= M - \DoF{\mathcal{R}_{[K]}^-}.
 \label{eqn:DoF_OIA}
\end{align}
Thus, we minimize the DoF loss caused by interference,
$\lim_{P\to\infty} (\mathcal{R}_{[K]}^-/\log_2P)$.
In next two subsections, we propose the OIA schemes to reduce the
DoF loss coming from the interferences.

%increase the achievable DoF by reducing the interferences'
%dimension,

%%%%%%%%%%%%%%%%%%%%%%%%%%%%%%%%%%%%%%%%%%%%%%%%%%%%%%%%%%%%%%%%%%%
% % % % % % % % % % % % % % % % % % % % % % % % % % % % % % % % % %
%%%%%%%%%%%%%%%%%%%%%%%%%%%%%%%%%%%%%%%%%%%%%%%%%%%%%%%%%%%%%%%%%%%
\subsection{OIA via Rate Loss Minimization (OIA1)}

Firstly, we directly minimize the average rate loss term at the
selected user via the postprocessing matrix design and user
selection. In this case, the average rate loss term becomes
\begin{align}
 \mathbb{E}_{\mathbf{H}}\big[
   \min_{k, \mathbf{F}_k}~\rate_k^- \big]
 =\mathbb{E}_{\mathbf{H}}\bigg[\underset{k, \mathbf{F}_k}{\min} ~\log_2\bigg\vert \mathbf{I}_M +
    \frac{P}{M}\sum_{i=2}^{3} \mathbf{F}_k \mathbf{H}_{k,i}
    \mathbf{H}_{k,i}^\JHdagger\mathbf{F}_k^\JHdagger \bigg\vert \bigg].
    \label{eqn:avr_rate_loss}
\end{align}
For each channel realization, the user $k$ minimizes the rate loss
term by using the postprocessing matrix given by
\begin{align}
 \mathbf{F}_k ^\OIA
 &\triangleq  \argmin{\mathbf{F}_k} ~\log_2\bigg\vert \mathbf{I}_M +
    \frac{P}{M}\sum_{i=2}^{3} \mathbf{F}_k \mathbf{H}_{k,i}
    \mathbf{H}_{k,i}^\JHdagger\mathbf{F}_k^\JHdagger \bigg\vert \NNL
 &=\argmin{\mathbf{F}_k}~
    \big\vert\mathbf{F}_k(\mathbf{H}_{k,2} \mathbf{H}_{k,2} ^\JHdagger +
    \mathbf{H}_{k,3} \mathbf{H}_{k,3}^\JHdagger )\mathbf{F}_k^\JHdagger\big\vert
    \NNL
 &= \left[ \mathbf{v}_{M+1} (\mathbf{B}_k),\ldots,
   \mathbf{v}_{2M} (\mathbf{B}_k) \right] ^\JHdagger,
   \label{eqn:OIA_FB_F}
\end{align}
where $\mathbf{B}_k= \mathbf{H}_{k,2} \mathbf{H}_{k,2} ^\JHdagger +
\mathbf{H}_{k,3} \mathbf{H}_{k,3} ^\JHdagger$, and the corresponding
rate loss term becomes $\log_2\prod_{m=M+1}^{2M}\big(1 + \frac{P}{M}
\lambda_{m} \left(\mathbf{B}_k \right) \big)$.
Thus, the required feedback information for the $k$th user becomes
\begin{align}
   \prod_{m=M+1}^{2M}\left(1 +  \frac{P}{M} \lambda_{m}
   \left(\mathbf{B}_k \right) \right),
   \label{eqn:feedback_OIA1}
\end{align}
and the selected user at the transmitter denoted by
$k_{\OIAI}^\star$ becomes
\begin{align}
 k_{\OIAI}^\star &=  \argmin{k}
 \prod_{m=M+1}^{2M}\bigg(1 +  \frac{P}{M} \lambda_{m}
    \left(\mathbf{B}_k \right) \bigg).
\end{align}

%%%%%%%%%%%%%%%%%%%%%%%%%%%%%%%%%%%%%%%%%%%%%%%%%%%%%%%%%%%%%%%%%%%
% % % % % % % % % % % % % % % % % % % % % % % % % % % % % % % % % %
%%%%%%%%%%%%%%%%%%%%%%%%%%%%%%%%%%%%%%%%%%%%%%%%%%%%%%%%%%%%%%%%%%%
\subsection{OIA via Chordal Distance Minimization (OIA2)}

As an alternative implementation, the transmitter can select a user
whose interfering channels are closest. The chordal distance is used
to measure the distance between the interfering channels at each
user.
%
%In this subsection, we propose another OIA scheme to minimize the
%chordal distance between interfering channels at the selected user.
%
Firstly, we find the upper bound of \eqref{eqn:avr_rate_loss} in the
following lemma.

%%%%%%%%%%%%%%%%%%%%%%%%%%%%%%%%%%%%%%%%%%%%%%%%%%%%%%%%%%%%%%%%%%%
% % % % % % % % % % % % % % % % % % % % % % % % % % % % % % % % % %
%%%%%%%%%%%%%%%%%%%%%%%%%%%%%%%%%%%%%%%%%%%%%%%%%%%%%%%%%%%%%%%%%%%
\begin{lemma}\label{lemma:rate_loss_bound}
The minimized average rate loss term given in
\eqref{eqn:avr_rate_loss} is upper bounded by
\begin{align}
 \mathbb{E}_{\mathbf{H}}\big[
   \min_{k, \mathbf{F}_k}~\rate_k^- \big]
 \le\mathbb{E}_{\tilde{\mathbf{H}}}
   \bigg\{\min_k~
   M\log_2\bigg[ 1 +  \frac{P}{M}
         d_c^2(\tilde{\mathbf{H}}_{k,2}, \tilde{\mathbf{H}}_{k,3})
         \bigg]
   \bigg\}, \label{eqn:rate_loss_bound}
\end{align}
where $\tilde{\mathbf{H}}_{k,i} \in \mathbb{C}^{N_R\times M}$ is an
arbitrary generator matrix of the subspace spanned by
$\mathbf{H}_{k,i}$.
\end{lemma}
%%%%%%%%%%%%%%%%%%%%%%%%%%%%%%%%%%%%%%%%%%%%%%%%%%%%%%%%%%%%%%%%%%%
% % % % % % % % % % % % % % % % % % % % % % % % % % % % % % % % % %
%%%%%%%%%%%%%%%%%%%%%%%%%%%%%%%%%%%%%%%%%%%%%%%%%%%%%%%%%%%%%%%%%%%
\begin{proof}
Since $\mathbf{H}_{k,i} \in \mathbb{C}^{N_R\times M}$, the matrix
$\mathbf{H}_{k,i} \mathbf{H}_{k,i}^\JHdagger$ has $M$ non-zero
eigenvalues. Thus, it can be decomposed by
\begin{align}
 \mathbf{H}_{k,i} \mathbf{H}_{k,i}^\JHdagger
   = \mathbf{U}_{k,i} \mathbf{\Lambda}_{k,i}
      \mathbf{U}_{k,i}^\JHdagger,
\end{align}
where $\mathbf{\Lambda}_{k,i} \in \mathbb{C}^{M \times M}$ is a
diagonal matrix whose diagonal elements are the non-zero eigenvalues
of $\mathbf{H}_{k,i} \mathbf{H}_{k,i}^\JHdagger$, and
$\mathbf{U}_{k,i} \in \mathbb{C}^{N_R \times M}$ consists of the
corresponding eigenvectors to the non-zero eigenvalues which becomes
a semi-orthogonal matrix such that $\mathbf{U}_{k,i}^\JHdagger
\mathbf{U}_{k,i} =\mathbf{I}_M$ but $\mathbf{U}_{k,i}
\mathbf{U}_{k,i}^\JHdagger \ne \mathbf{I}_{N_R}$
\footnote{Sometimes this decomposition is referred to compact SVD or
thin SVD.}.
Using this decomposition, we can bound \eqref{eqn:avr_rate_loss} as
follows:
\begin{align}
 \mathbb{E}_{\mathbf{H}}\big[
   \min_{k, \mathbf{F}_k}~\rate_k^- \big]
 &\stackrel{(a)}{=}\mathbb{E}_{\mathbf{U}}\Big\{~
   \mathbb{E}_{\mathbf{\Lambda}}\big[
   \min_{k, \mathbf{F}_k}~\rate_k^- \big]~\Big\}\NNL
 &\stackrel{(b)}{\le}\mathbb{E}_{\mathbf{U}}\Big\{~
   \min_{k, \mathbf{F}_k}~
   \mathbb{E}_{\mathbf{\Lambda}}\big[\rate_k^- \big] ~\Big\}\NNL
 &\stackrel{(c)}{\le}\mathbb{E}_{\mathbf{U}}
   \bigg\{\min_k~
   M\log_2\bigg[ 1 +  \frac{P}{M}
   d_c^2(\mathbf{U}_{k,2}, \mathbf{U}_{k,3})
   \bigg]\bigg\},\label{eqn:rate_loss_bound2}
\end{align}
where the equality $(a)$ holds from the fact that $\mathbf{U}_{k,i}$
and $\mathbf{\Lambda}_{k,i}$ are independent of each other
\cite{RJ2008}, and the inequality $(b)$ is because the average of
the minimum values is smaller than the minimum of the average
values.
The inequality $(c)$ holds because
\begin{align}
 \min_{k, \mathbf{F}_k}~
   \mathbb{E}_{\mathbf{\Lambda}}\big[\rate_k^- \big]
 &= \min_{k, \mathbf{F}_k}~
 \mathbb{E}_{\mathbf{\Lambda}}
    \log_2\bigg\vert\mathbf{I}_M +
    \frac{P}{M} \mathbf{F}_k \Big( \sum_{i=2}^{3}
    \mathbf{U}_{k,i}\mathbf{\Lambda}_{k,i}
    \mathbf{U}_{k,i}^\JHdagger \Big) \mathbf{F}_k^\JHdagger
    \bigg\vert        \NNL
 &\stackrel{(c_1)}{\le}  \min_{k, \mathbf{F}_k}~
    \log_2\bigg\vert\mathbf{I}_M +
    P \mathbf{F}_k \Big( \sum_{i=2}^{3}
    \mathbf{U}_{k,i} \mathbf{U}_{k,i}^\JHdagger \Big) \mathbf{F}_k^\JHdagger
    \bigg\vert \NNL
 &\stackrel{(c_2)}{=}\min_k~
 \log_2\prod_{m=M+1}^{2M}\bigg(1 +  P \lambda_{m}
    \left(\mathbf{C}_k \right) \bigg)\NNL
 &\stackrel{(c_3)}{\le} \min_k~M\log_2\bigg[ 1 +  \frac{P}{M}
    \sum_{m=M+1}^{2M}\lambda_{m}\left(\mathbf{C}_k \right) \bigg]
    \NNL
 &\stackrel{(c_4)}{=} \min_k~
   M\log_2\bigg[ 1 +  \frac{P}{M}
   d_c^2(\mathbf{U}_{k,2}, \mathbf{U}_{k,3}) \bigg],\NN
\end{align}
where $\mathbf{C}_k = \sum_{i=2}^{3} \mathbf{U}_{k,i}
\mathbf{U}_{k,i}^\JHdagger$.
The inequality $(c_1)$ is from the Jensen's inequality and
$\mathbb{E} [\mathbf{\Lambda}_{k,i}] = M\mathbf{I}_M$ \cite{RJ2008}.
The equality $(c_2)$ is obtained by applying $\mathbf{F}_k = \left[
\mathbf{v}_{M+1} (\mathbf{C}_k),\ldots, \mathbf{v}_{2M}
(\mathbf{C}_k) \right] ^\JHdagger$. Also, the inequality $(c_3)$ is
from the concavity of a logarithm function with the Jensen's
inequality. Finally, the equality $(c_4)$ is satisfied from Lemma
\ref{lemma:p_angles_sum}.
Although $\mathbf{U}_{k,i}$ is one of the generator matrices of the
subspace formed by $\mathbf{H}_{k,i}$, it can be replaced by any
arbitrary generator matrices denoted by $\tilde{\mathbf{H}}_{k,i}$
because the chordal distance is uniquely defined for any generator
matrices. Thus, the bound \eqref{eqn:rate_loss_bound2} can by
equivalently rewritten by \eqref{eqn:rate_loss_bound}.
\end{proof}
%%%%%%%%%%%%%%%%%%%%%%%%%%%%%%%%%%%%%%%%%%%%%%%%%%%%%%%%%%%%%%%%%%%
% % % % % % % % % % % % % % % % % % % % % % % % % % % % % % % % % %
%%%%%%%%%%%%%%%%%%%%%%%%%%%%%%%%%%%%%%%%%%%%%%%%%%%%%%%%%%%%%%%%%%%

In OIA2, we minimize \eqref{eqn:rate_loss_bound} instead of
\eqref{eqn:avr_rate_loss}.
Thus, the feedback information at user $k$  becomes $d_c^2
(\tilde{\mathbf{H}}_{k,2}, \tilde{\mathbf{H}}_{k,3})$ given by
\begin{align}
 d_c^2(\tilde{\mathbf{H}}_{k,2}, \tilde{\mathbf{H}}_{k,3})
   &=\frac{1}{2}\Vert \tilde{\mathbf{H}}_{k,2}\tilde{\mathbf{H}}_{k,2}^\JHdagger
        - \tilde{\mathbf{H}}_{k,3}\tilde{\mathbf{H}}_{k,3}^\JHdagger \Vert_F \NNL
   &= M-tr(\tilde{\mathbf{H}}_{k,2}^\JHdagger \tilde{\mathbf{H}}_{k,3}
   \tilde{\mathbf{H}}_{k,3}^\JHdagger \tilde{\mathbf{H}}_{k,2}),
   \label{eqn:feedback_OIA2}
\end{align}
and the index of the selected user denoted by $k_{\OIAII}^\star$
becomes
\begin{align}
 k_\OIAII^\star
   = \argmin{k} ~ d_c^2(\tilde{\mathbf{H}}_{k,2},
   \tilde{\mathbf{H}}_{k,3}).
   \label{eqn:k1_OIA}
\end{align}

%%%%%%%%%%%%%%%%%%%%%%%%%%%%%%%%%%%%%%%%%%%%%%%%%%%%%%%%%%%%%%%%%%%
% % % % % % % % % % % % % % % % % % % % % % % % % % % % % % % % % %
%%%%%%%%%%%%%%%%%%%%%%%%%%%%%%%%%%%%%%%%%%%%%%%%%%%%%%%%%%%%%%%%%%%
\begin{remark}
In OIA1, each user requires SVD to find the feedback information
\eqref{eqn:feedback_OIA1}, and concurrently the postprocessing
matrix is obtained.
In OIA2, however, each user only needs to find the generator
matrices of the interfering channels for the feedback information
given in \eqref{eqn:feedback_OIA2}.
Although the generator matrix can be obtained by various ways such
as SVD and QR decomposition, each user adopts the QR decomposition
to find the generator matrix since it is simpler than SVD.
Thus, we can greatly reduce the computational complexity of OIA2
compared with OIA1. We describe details on this in Section
\ref{subsection:complexities}.
\end{remark}
%%%%%%%%%%%%%%%%%%%%%%%%%%%%%%%%%%%%%%%%%%%%%%%%%%%%%%%%%%%%%%%%%%%
% % % % % % % % % % % % % % % % % % % % % % % % % % % % % % % % % %
%%%%%%%%%%%%%%%%%%%%%%%%%%%%%%%%%%%%%%%%%%%%%%%%%%%%%%%%%%%%%%%%%%%

To quantify the rate loss at the selected user, we should find the
relationship between feedback value from the selected user and the
number of total users, i.e., the relationship between
$\mathbb{E}\Big[\underset{1\le k\le K}{\min} d_c^2(
\tilde{\mathbf{H}}_{k,2}, \tilde{\mathbf{H}}_{k,3})\Big]$ and $K$.
The following lemma helps us to obtain the average feedback value
from the selected user.

%%The average feedback information at the selected user, i.e.,
%$\AVR{\min_k~d_c^2(\tilde{\mathbf{H}}_{k,2},
%\tilde{\mathbf{H}}_{k,3})}$, is obtained from the following
%lemma.
%
%\begin{align}
% D \triangleq \AVR{\min_k~d_c^2(\tilde{\mathbf{H}}_{k,2},
%   \tilde{\mathbf{H}}_{k,3})},
%   \label{eqn:D}
%\end{align}
%
%it can be obtained by following lemma.

%%%%%%%%%%%%%%%%%%%%%%%%%%%%%%%%%%%%%%%%%%%%%%%%%%%%%%%%%%%%%%%%%%%
% % % % % % % % % % % % % % % % % % % % % % % % % % % % % % % % % %
%%%%%%%%%%%%%%%%%%%%%%%%%%%%%%%%%%%%%%%%%%%%%%%%%%%%%%%%%%%%%%%%%%%
\begin{lemma}\label{lemma:distortion_bound}
The average feedback value from the selected user is equivalent to
the average of the minimum chordal distance when we quantize an
arbitrary subspace $\mathcal{A} \in \mathcal{G}_{N_R,M}(\mathbb{C})$
with one of the $K$ random subspaces $\mathcal{C}_{\mathrm{rnd}}
\subset \mathcal{G}_{N_R,M} (\mathbb{C})$ such that
\begin{align}
 \AVR{\min_k~d_c^2(\tilde{\mathbf{H}}_{k,2},
    \tilde{\mathbf{H}}_{k,3})}
 =  \mathbb{E}_{\mathcal{C}_{\mathrm{rnd}}}\left[
 \min_{\mathbf{W} \in \mathcal{C}_{\mathrm{rnd}}}
    d_c^2(\mathbf{A}, \mathbf{W})\right].
 \label{eqn:distortion}
\end{align}
\end{lemma}
%%%%%%%%%%%%%%%%%%%%%%%%%%%%%%%%%%%%%%%%%%%%%%%%%%%%%%%%%%%%%%%%%%%
% % % % % % % % % % % % % % % % % % % % % % % % % % % % % % % % % %
%%%%%%%%%%%%%%%%%%%%%%%%%%%%%%%%%%%%%%%%%%%%%%%%%%%%%%%%%%%%%%%%%%%

\begin{proof}
Consider an arbitrary subspace $\mathcal{A}\in\mathcal{G}_{N_R,
M}(\mathbb{C})$ and its generator matrix
$\mathbf{A}\in\mathbb{C}^{N_R\times M}$.
Then, we define the rotation matrix $\mathbf{R}_k \in
\mathbb{C}^{N_R\times N_R}$ at the $k$th user, which rotates
$\tilde{\mathbf{H}}_{k,2}$ to $\mathbf{A}$ such that $\mathbf{R}_k
\tilde{\mathbf{H}} _{k,2} = \mathbf{A}$.
If we denote the generator matrix of the null space of $\mathcal{A}$
by $\mathbf{A} ^\perp \in \mathbb{C}^{N_R\times M}$, the matrix
$\mathbf{R}_k$ can be represented by
\begin{align}
 \mathbf{R}_k =
    \left[\mathbf{A}, \mathbf{A}^\perp\right]
    \left[\tilde{\mathbf{H}} _{k,2}, \tilde{\mathbf{H}} _{k,2}
    ^\perp\right]^\JHdagger,
\end{align}
which becomes a unitary matrix, i.e., $\mathbf{R}_k^\JHdagger
\mathbf{R}_k = \mathbf{R}_k \mathbf{R}_k^\JHdagger =
\mathbf{I}_{N_R}$.
Since the chordal distance is invariant with a rotation, the chordal
distance at the $k$th user satisfies
\begin{align}
 d_c^2(\tilde{\mathbf{H}}_{k,2}, \tilde{\mathbf{H}}_{k,3})
 = d_c^2(\mathbf{R}_k\tilde{\mathbf{H}}_{k,2},
    \mathbf{R}_k\tilde{\mathbf{H}}_{k,3})
 = d_c^2(\mathbf{A}, \mathbf{R}_k\tilde{\mathbf{H}}_{k,3}).
\end{align}
The chordal distance at the selected user becomes
\begin{align}
% d_c^2(\tilde{\mathbf{H}}_{k_2^\star,2},
%     \mathbf{H}_{k_2^\star,3})
 \min_k~d_c^2(\tilde{\mathbf{H}}_{k,2}, \tilde{\mathbf{H}}_{k,3})
 &=\min_k~d_c^2(\mathbf{R}_k\tilde{\mathbf{H}}_{k,2},
    \mathbf{R}_k\tilde{\mathbf{H}}_{k,3})\NNL
 &=\min_{\mathbf{W} \in \mathcal{C}_{\mathrm{rnd}}}
    ~d_c^2(\mathbf{A}, \mathbf{W})
\end{align}
where $\mathcal{C}_{\mathrm{rnd}} \subset \mathcal{G}_{N_R, M} $ is
a set of $K$ random subspaces such that $\mathcal{C}_{\mathrm{rnd}}
= \{\mathbf{R}_k \tilde{\mathbf{H}}_{k,3}\} _{k=1}^K$.
Thus, the average chordal distance at the selected user can be given
by the average of the minimum chordal distance between an arbitrary
subspace and its quantized subspace by one of the $K$ random
subspaces as in \eqref{eqn:distortion}.
\end{proof}
%%%%%%%%%%%%%%%%%%%%%%%%%%%%%%%%%%%%%%%%%%%%%%%%%%%%%%%%%%%%%%%%%%%
% % % % % % % % % % % % % % % % % % % % % % % % % % % % % % % % % %
%%%%%%%%%%%%%%%%%%%%%%%%%%%%%%%%%%%%%%%%%%%%%%%%%%%%%%%%%%%%%%%%%%%

It has been shown that the average quantization error when an
arbitrary source on the Grassmann manifold $\mathcal{G}_{N_R, M}
(\mathbb{C})$ is quantized with the random codebook
$\mathcal{C}_{\mathrm{rnd}} \subset \mathcal{G}_{N_R, M}
(\mathbb{C})$ of size $K$ is upper bounded by $D$ \cite{DLR2008},
i.e.,
\begin{align}
\mathbb{E} \Big[ \min_{ \mathbf{W} \in
    \mathcal{C}_{\mathrm{rnd}}}
    d_c^2 (\mathbf{H}, \mathbf{W}) \Big] \le D,
    \label{eqn:QE_bound}
\end{align}
where $D$ is given by
\begin{align}
 D=
    &\frac{\Gamma\left(\frac{1}{M^2}\right)}{M^2}
    (
        \eta K
    )^{-\frac{1}{M^2}}
    +M\exp\left[
        -\left(\eta K\right)^{1-a}
    \right]
    \label{eqn:D_bar}
\end{align}
with $\eta = \frac{1}{\Gamma(M^2+1)} \prod_{i=1} ^{M}
\frac{\Gamma(2M-i+1)} {\Gamma(M-i+1)}$,
and $a\in (0,1)$ is a real number chosen to satisfy $(\eta
K)^{\frac{-a}{M^2}}\le 1$.
Thus, from Lemma \ref{lemma:distortion_bound} and
\eqref{eqn:QE_bound}, we can conclude that the average feedback
value from the selected user is upper bounded as
\begin{align}
 \AVR{\min_k~d_c^2(\tilde{\mathbf{H}}_{k,2},
    \tilde{\mathbf{H}}_{k,3})} \le D. \label{eqn:D_le_Dbar}
\end{align}
Note that the second term in \eqref{eqn:D_bar} can be negligible
compared to the first term for large $K$\cite{DLR2008}, and the main
order term of \eqref{eqn:D_bar} is sufficiently
accurate\cite{DLR2008, DLLR2009, RJ2008}.

Once a user is selected at the transmitter, the selected user only
finds the postprocessing matrix to minimize the rate loss term,
which is given in \eqref{eqn:OIA_FB_F}, i.e., only the user
$k_2^\star$ finds the postprocessing matrix $\mathbf{F}
_{k_2^\star}^\OIA$.

%%%%%%%%%%%%%%%%%%%%%%%%%%%%%%%%%%%%%%%%%%%%%%%%%%%%%%%%%%%%%%%%%%%
% % % % % % % % % % % % % % % % % % % % % % % % % % % % % % % % % %
%%%%%%%%%%%%%%%%%%%%%%%%%%%%%%%%%%%%%%%%%%%%%%%%%%%%%%%%%%%%%%%%%%%
\JHnewpage\section{Achievable Rate and Degrees of Freedom (DoF)}%: $N_R=2M$}
This section analyzes the achievable rate of the proposed OIA
schemes and their DoF. Without loss of generality, the average
achievable rate and a DoF at the first transmitter are derived as in
the previous section.
We start from the following lemma.

%%%%%%%%%%%%%%%%%%%%%%%%%%%%%%%%%%%%%%%%%%%%%%%%%%%%%%%%%%%%%%%%%%%
% % % % % % % % % % % % % % % % % % % % % % % % % % % % % % % % % %
%%%%%%%%%%%%%%%%%%%%%%%%%%%%%%%%%%%%%%%%%%%%%%%%%%%%%%%%%%%%%%%%%%%
\begin{lemma}\label{lemma:M-alphaM}
When the number of users (i.e., $K$) is fixed and invariant to $P$,
the achievable DoF by the proposed OIA schemes becomes zero such
that
\begin{align}
 \lim_{P\to\infty \atop \textrm{\normalfont Fixed } K }
 \frac{\mathcal{R}_{[K]}}{\log_2 P} = 0. \NN
\end{align}
\end{lemma}
%%%%%%%%%%%%%%%%%%%%%%%%%%%%%%%%%%%%%%%%%%%%%%%%%%%%%%%%%%%%%%%%%%%
\begin{proof}
We can directly derive the achievable DoF from \eqref{eqn:DoF_OIA}.
At the user $k$, the matrix $\sum_{i=2}^3\mathbf{H}_{k, i}
\mathbf{H}_{k,i} ^\JHdagger$ has $2M$ non-zero eigenvalues with
probability one.
At the selected user $k^\star$ ($k^\star = k_1^\star$ or $k_2^\star$
using OIA1 or OIA2), the matrix $\sum_{i=2}^3 \mathbf{H}_{k^\star,
i} \mathbf{H} _{k^\star,i} ^\JHdagger$ also has $2M$ eigenvalues so
that $\sum_{i=2}^3 \mathbf{F}_{k^\star}^\OIA \mathbf{H}_{k^\star, i}
\mathbf{H} _{k^\star,i} ^\JHdagger\mathbf{F}_{k^\star}
^{\OIA\JHdagger}$ becomes a full rank matrix having $M$ non-zero
eigenvalues.
Thus, when $K$ is fixed (invariant with $P$), one can easily find
that $\underset{P\to\infty}{\lim} \frac{\mathcal{R}_{[K]} ^-}{\log_2
P} = M$.
Substituting this into \eqref{eqn:DoF_OIA}, we complete the proof.
%
%~\\~\\~\\~\\
%When $K$ is fixed, it is satisfied that $\underset{P\to\infty}{\lim}
%\frac{\mathcal{R}_{[K]} ^-}{\log_2 P} = M$
%
%because at the selected user $k^\star$ ($k^\star = k_1^\star$ or
%$k_2^\star$ using OIA1 or OIA2)
%
%the matrix $\mathbf{H}_{k^\star, 2} \mathbf{H}_{k^\star, 2}^\JHdagger
%+ \mathbf{H}_{k^\star, 3} \mathbf{H}_{k^\star, 3}^\JHdagger$ has
%non-zero $2M$ eigenvalues under the finite number of users with
%probability one. Thus, \eqref{eqn:OIA_DoF_min} holds.
%
\end{proof}
\vspace{.1in}
%%%%%%%%%%%%%%%%%%%%%%%%%%%%%%%%%%%%%%%%%%%%%%%%%%%%%%%%%%%%%%%%%%%
% % % % % % % % % % % % % % % % % % % % % % % % % % % % % % % % % %
%%%%%%%%%%%%%%%%%%%%%%%%%%%%%%%%%%%%%%%%%%%%%%%%%%%%%%%%%%%%%%%%%%%
%\begin{figure}[!t]
%\centering
%  \includegraphics[width=\columnwidth]{OIA_10users_varying_alpha_M2.eps}\\
%  \caption{The achievable rate per user of OIA scheme varying $\alpha$ when $(N_T,M,N_R)=(2,2,4)$.}
%  \label{fig:varying_alpha}
%\end{figure}

Fig. \ref{fig:varying_K} shows the average achievable rates of each
user with the proposed OIA2 scheme for $K=10$ and $K=50$,
respectively, when $(N_T, M, N_R)=(2, 2, 4)$. As stated in Lemma
\ref{lemma:M-alphaM}, the achievable DoF of each user becomes always
zero when the number of users is finite.

On the other hand, by increasing the number of users we can reduce
the rate loss term so that the positive DoF can be obtained at the
first transmitter.
In the next lemma, we find the upper bound of the rate loss term as
a function of the number of users.

%%%%%%%%%%%%%%%%%%%%%%%%%%%%%%%%%%%%%%%%%%%%%%%%%%%%%%%%%%%%%%%%%%%
% % % % % % % % % % % % % % % % % % % % % % % % % % % % % % % % % %
%%%%%%%%%%%%%%%%%%%%%%%%%%%%%%%%%%%%%%%%%%%%%%%%%%%%%%%%%%%%%%%%%%%
\begin{lemma}\label{lemma:capacity_loss_bound}
When the first user group has $K$ users, the average rate loss term
at the selected user is bounded by
\begin{align}
 \mathcal{R}_{[K]}^{\textrm{loss}}\le M\log_2\left(1+\frac{P}{M}D
    \right),\label{eqn:R_loss_bound}
\end{align}
where $D$ is given in \eqref{eqn:D_bar}.
\end{lemma}
%%%%%%%%%%%%%%%%%%%%%%%%%%%%%%%%%%%%%%%%%%%%%%%%%%%%%%%%%%%%%%%%%%%
\begin{proof}
The inequality \eqref{eqn:rate_loss_bound} in Lemma
\ref{lemma:rate_loss_bound} can be further bounded by
\begin{align}
 \mathbb{E}_{\mathbf{H}}\big[
   \min_{k, \mathbf{F}_k}~R_k^- \big]
 &\le\mathbb{E}_{\tilde{\mathbf{H}}}
   \bigg\{\min_k~
   M\log_2\bigg[ 1 +  \frac{P}{M}
         d_c^2(\tilde{\mathbf{H}}_{k,2}, \tilde{\mathbf{H}}_{k,3})
         \bigg] \bigg\}\NNL
 &=\mathbb{E}_{\tilde{\mathbf{H}}}
   \bigg\{
   M\log_2\bigg[ 1 +  \frac{P}{M}\Big[
         \min_k d_c^2(\tilde{\mathbf{H}}_{k,2}, \tilde{\mathbf{H}}_{k,3})
         \Big]\bigg] \bigg\}\NNL
 &\stackrel{(a)}{\le}
   M\log_2\bigg[ 1 +  \frac{P}{M}
         \mathbb{E}_{\tilde{\mathbf{H}}}\left[ \min_k~
         d_c^2(\tilde{\mathbf{H}}_{k,2}, \tilde{\mathbf{H}}_{k,3})
         \right] \bigg] \NNL
 &\stackrel{(b)}{\le}
   M\log_2\left( 1 +  \frac{P}{M}
         D \right)\NN
\end{align}
where the inequality $(a)$ is from the Jensen's inequality, and the
inequality $(b)$ is from \eqref{eqn:D_le_Dbar}.
\end{proof}\vspace{.1in}
%%%%%%%%%%%%%%%%%%%%%%%%%%%%%%%%%%%%%%%%%%%%%%%%%%%%%%%%%%%%%%%%%%%

%%%%%%%%%%%%%%%%%%%%%%%%%%%%%%%%%%%%%%%%%%%%%%%%%%%%%%%%%%%%%%%%%%%
% % % % % % % % % % % % % % % % % % % % % % % % % % % % % % % % % %
%%%%%%%%%%%%%%%%%%%%%%%%%%%%%%%%%%%%%%%%%%%%%%%%%%%%%%%%%%%%%%%%%%%
\begin{theorem}\label{theorem:OIA_infinite_K}
When the transmit power is fixed and the number of users goes to
infinity, i.e., $K\to\infty$, the achievable rate at the selected
user becomes the ergodic capacity of the $M\times M$ point-to-point
MIMO system without interference.
\end{theorem}
%%%%%%%%%%%%%%%%%%%%%%%%%%%%%%%%%%%%%%%%%%%%%%%%%%%%%%%%%%%%%%%%%%%
% % % % % % % % % % % % % % % % % % % % % % % % % % % % % % % % % %
%%%%%%%%%%%%%%%%%%%%%%%%%%%%%%%%%%%%%%%%%%%%%%%%%%%%%%%%%%%%%%%%%%%
\begin{proof}
When the transmit power is fixed, the rate loss term becomes zero as
the number of users goes to infinity. This can be obtained from
Lemma 6 using $\lim_{K\to\infty} D = 0$.
Thus, when the number of users increases and the transmit power is
fixed, the achievable rate using OIA2 becomes
\begin{align}
 \lim_{K\to\infty \atop \textrm{Fixed~} P} \mathcal{R}_{[K]}
 &= \mathbb{E}\log_2
    \bigg\vert
    \mathbf{I}_M + \frac{P}{M}
    \mathbf{F}_{k_\OIAII^\star}\mathbf{H}_{k_\OIAII^\star,1}
    \mathbf{H}_{k_\OIAII^\star,1}^\JHdagger \mathbf{F}_{k_\OIAII^\star}^\JHdagger
    \bigg\vert\NNL
 &= \mathbb{E}\log_2
    \bigg\vert
    \mathbf{I}_M + \frac{P}{M}
    \hat{\mathbf{H}} \hat{\mathbf{H}}^\JHdagger
    \bigg\vert,
    \label{eqn:MM_MIMO}
\end{align}
where $\hat{\mathbf{H}} \triangleq \mathbf{F}_{k_\OIAII^\star}
\mathbf{H}_{k_\OIAII^\star,1}$ becomes an $M\times M$ matrix whose
elements are i.i.d. Gaussian random variables with zero mean and
unit variance.
This is because $\mathbf{F}_{k_\OIAII^\star} \in \mathbb{C}^{M\times
N_R}$ is a semi-unitary matrix independently chosen on
$\mathbf{H}_{k_\OIAII^\star,1}$ such that
$\mathbf{F}_{k_\OIAII^\star} \mathbf{F}_{k_\OIAII^\star}^\JHdagger =
\mathbf{I}_M$.
The result in \eqref{eqn:MM_MIMO} implicates that when the number of
users goes to infinity, each transmitter achieves the same ergodic
rate as the ergodic capacity of an interference-free $M\times M$
point-to-point MIMO system.
Proof for the OIA1 case is trivial.
\end{proof}
%%%%%%%%%%%%%%%%%%%%%%%%%%%%%%%%%%%%%%%%%%%%%%%%%%%%%%%%%%%%%%%%%%%
% % % % % % % % % % % % % % % % % % % % % % % % % % % % % % % % % %
%%%%%%%%%%%%%%%%%%%%%%%%%%%%%%%%%%%%%%%%%%%%%%%%%%%%%%%%%%%%%%%%%%%
In Lemma \ref{lemma:M-alphaM}, we showed that the achievable DoF by
OIA becomes zero when the number of users is fixed. Theorem 1
implicates that the achievable rate by the proposed OIA schemes
becomes the same as that of an $M\times M$ point-to-point MIMO
system when the number of users increases under a fixed power. Based
on these results, we can conjecture that the achievable DoF by the
proposed OIA schemes will be ranged in $[0,M]$ if the number of
users is sufficiently large, i.e.,
\begin{align}
 \lim_{P\to\infty}
 \left[
    \lim_{K\to\infty}
    \frac{\mathcal{R}_{[K]}}{\log_2 P}
 \right] \in [0, M]. \NN
\end{align}
The increasing speeds of $K$ and $P$ will determine the value of
achievable DoF and Theorem 2 establishes the relationship between
achievable DoF and the required number of users.
%%%%%%%%%%%%%%%%%%%%%%%%%%%%%%%%%%%%%%%%%%%%%%%%%%%%%%%%%%%%%%%%%%%
% % % % % % % % % % % % % % % % % % % % % % % % % % % % % % % % % %
%%%%%%%%%%%%%%%%%%%%%%%%%%%%%%%%%%%%%%%%%%%%%%%%%%%%%%%%%%%%%%%%%%%
\begin{theorem}\label{theorem:required_K}
At each transmitter, DoF $m\in[0,M]$ is obtained when the number of
users is scaled as
\begin{align}
    K\propto P^{mM}.\nonumber
\end{align}
\end{theorem}
%%%%%%%%%%%%%%%%%%%%%%%%%%%%%%%%%%%%%%%%%%%%%%%%%%%%%%%%%%%%%%%%%%%
% % % % % % % % % % % % % % % % % % % % % % % % % % % % % % % % % %
%%%%%%%%%%%%%%%%%%%%%%%%%%%%%%%%%%%%%%%%%%%%%%%%%%%%%%%%%%%%%%%%%%%
\begin{proof}
Because the achievable DoF using OIA is given by $M -
\DoF{\mathcal{R}_{[K]}^-}$, the equivalent condition for DoF $m$ is
\begin{align}
 \DoF{\mathcal{R}_{[K]}^-} = M-m. \label{eqn:DoF_M-m}
\end{align}
Using the upper bound given in \eqref{eqn:R_loss_bound}, we obtain
the sufficient scaling for \eqref{eqn:DoF_M-m} such that
\begin{align}
 \DoF{M\log_2\left(1+\frac{P}{M} D \right)}=M-m.
\end{align}
Substituting \eqref{eqn:D_bar} into above equation, we obtain the
required user scaling $K \propto P^{mM}$ to obtain DoF of $m$ at
each transmitter.
\end{proof}

%%%%%%%%%%%%%%%%%%%%%%%%%%%%%%%%%%%%%%%%%%%%%%%%%%%%%%%%%%%%%%%%%%%
% % % % % % % % % % % % % % % % % % % % % % % % % % % % % % % % % %
%%%%%%%%%%%%%%%%%%%%%%%%%%%%%%%%%%%%%%%%%%%%%%%%%%%%%%%%%%%%%%%%%%%
\begin{remark}
In Theorem \ref{theorem:OIA_infinite_K}, we have shown that each
transmitter and the selected user communicate like an interference
free $M\times M$ MIMO system as the number of users goes to infinity
for fixed SNR.
Theorem \ref{theorem:required_K} implicates that the transmitter can
asymptotically achieve the same rate as the capacity of an
interference free $M\times M$ MIMO system with user scaled as $K
\propto P^{M^2}$ in high SNR region.
\end{remark}

In Fig. \ref{fig:user_scaled_alpha10}, the achievable rate per
transmitter with OIA2 scheme is plotted when $(N_T,M,N_R)=(1,1,2)$.
With the user scaling $K \propto P^{M^2}$, the achievable DoF is
maintained as $M$ as predicted in Theorem \ref{theorem:required_K}.

%%%%%%%%%%%%%%%%%%%%%%%%%%%%%%%%%%%%%%%%%%%%%%%%%%%%%%%%%%%%%%%%%%%
% % % % % % % % % % % % % % % % % % % % % % % % % % % % % % % % % %
%%%%%%%%%%%%%%%%%%%%%%%%%%%%%%%%%%%%%%%%%%%%%%%%%%%%%%%%%%%%%%%%%%%

%%%%%%%%%%%%%%%%%%%%%%%%%%%%%%%%%%%%%%%%%%%%%%%%%%%%%%%%%%%%%%%%%%%
% % % % % % % % % % % % % % % % % % % % % % % % % % % % % % % % % %
%%%%%%%%%%%%%%%%%%%%%%%%%%%%%%%%%%%%%%%%%%%%%%%%%%%%%%%%%%%%%%%%%%%
\JHnewpage\section{Comparison with Conventional Opportunistic User
Selection}

In this section, we compare the proposed OIA schemes with
conventional user selection schemes in terms of computational
complexities and achievable rate.

%OIA aligns interferences by selecting the user whose interfering
%channels are most aligned. This section compares OIA scheme with
%conventional opportunistic user selection schemes as well as the
%computational complexities.

%%%%%%%%%%%%%%%%%%%%%%%%%%%%%%%%%%%%%%%%%%%%%%%%%%%%%%%%%%%%%%%%%%%
% % % % % % % % % % % % % % % % % % % % % % % % % % % % % % % % % %
%%%%%%%%%%%%%%%%%%%%%%%%%%%%%%%%%%%%%%%%%%%%%%%%%%%%%%%%%%%%%%%%%%%
\subsection{Maximum SNR User Selection (MAX-SNR)}
Firstly, we consider the maximum SNR user selection scheme
(MAX-SNR).
In this scheme, each user maximizes the achievable rate ignoring the
effects of the interfering channels.
At the $k$th user, the postprocessing matrix is designed by
\begin{align}
 \mathbf{F}_k^\SNR
 &\triangleq \argmax{\mathbf{F}_k}~
    \log_2  \left\vert
    \mathbf{I}_M + \frac{P}{M}\mathbf{F}_k
    \mathbf{H}_{k,1} \mathbf{H}_{k,1}^\JHdagger
    \mathbf{F}_k^\JHdagger \right\vert,
    \label{eqn:postprocessing_SNR}
\end{align}
and thus $\mathbf{F}_{k}^\SNR =\left[\mathbf{v}_{1}
(\mathbf{A}_k),\ldots,
\mathbf{v}_{M}(\mathbf{A}_k)\right]^\JHdagger$
where $\mathbf{A}_k = \mathbf{H}_{k,1} \mathbf{H}_{k,1}^\JHdagger$.
The corresponding achievable rate at the $k$th user becomes $\log_2
\prod_{m=1}^{M}\left( 1+\frac{P}{M} \lambda_m( \mathbf{H}_{k,1}
\mathbf{H}_{k,1}^\JHdagger )\right)$.
Thus, the feedback information from the $k$th user becomes
$\prod_{m=1}^{M}\left( 1+\frac{P}{M} \lambda_m( \mathbf{H}_{k,1}
\mathbf{H}_{k,1}^\JHdagger )\right)$, and the index of the selected
user denoted by $k_\SNR^\star$ becomes
\begin{align}
 k_\SNR^\star
    = \argmax{k} ~ \prod_{m=1}^{M} \left(
    1+\frac{P}{M} \lambda_m (\mathbf{H}_{k,1}
    \mathbf{H}_{k,1}^\JHdagger)\right).
\end{align}
\subsection{Time Division Multiplexing}
In this subsection, we consider two time division multiplexing
schemes.
In the first time division multiplexing scheme (TDM1), only one of
the three transmitters serves its selected user at any time
instance. Therefore, the selected user does not receive any
interference from other transmitters.
Each user finds the postprocessing matrix to maximize the achievable
rate, so the postprocessing matrix at the transmitter is the same as
that of the MAX-SNR scheme given in \eqref{eqn:postprocessing_SNR}.
Also, the feedback information from each user and the user selection
criterion are exactly the same as those of the MAX-SNR scheme.
Because only one selected user is exclusively served by the TDM
approach, the achievable DoF per transmitter becomes $M/3$.

We also consider another time division multiplexing scheme (TDM2)
where only two of three transmitters serve their selected users.
Since each user has $2M$ antennas, three transmitters can obtain
$2M$ DoF for each channel realization, i.e., each transmitter can
achieve $2M/3$ DoF. In TDM2, each transmitter selects a user who has
the minimum rate loss term. When the first and the second
transmitters simultaneously transmit, the rate loss term of the
$k$th user of the first transmitter is minimized as
\begin{align}
 \min_{\mathbf{F}_k}\log_2\left\vert \mathbf{I}_M + \frac{P}{M} \mathbf{F}_k \mathbf{H}_{k,2}
    \mathbf{H} _{k,2} ^\JHdagger \mathbf{F}_k ^\JHdagger \right\vert
 =\prod_{m=M+1}^{2M}\left(1 +  \frac{P}{M} \lambda_{m} \big(
    \mathbf{H}_{k,2} \mathbf{H} _{k,2} ^\JHdagger \big) \right).\NN
%    \label{eqn:FBinfo_TDM2A}
\end{align}
Therefore, the required feedback information at the $k$th user
becomes the right-hand-side of the equality, and the selected user
at the first transmitter denoted by $k_{\TDMII}^\star$ becomes
\begin{align}
 k_{\TDMII}^\star &=  \argmin{k}
 \prod_{m=M+1}^{2M}\bigg(1 +  \frac{P}{M} \lambda_{m}
    \left(\mathbf{H}_{k,2} \mathbf{H} _{k,2} ^\JHdagger \right)
    \bigg).\NN
\end{align}

%%%%%%%%%%%%%%%%%%%%%%%%%%%%%%%%%%%%%%%%%%%%%%%%%%%%%%%%%%%%%%%%%%%
% % % % % % % % % % % % % % % % % % % % % % % % % % % % % % % % % %
%%%%%%%%%%%%%%%%%%%%%%%%%%%%%%%%%%%%%%%%%%%%%%%%%%%%%%%%%%%%%%%%%%%
\subsection{Complexity Analysis}\label{subsection:complexities}

In this subsection, the computational complexity of each scheme is
represented by the number of floating point operations (flops)
\cite[Chap. 1]{GL1989}. An addition, multiplication, or division of
real numbers is counted as one flop, so a complex addition and
multiplication are counted as two flops and six flops, respectively.
For an $m\times n$ complex matrix $\mathbf{G}\in\mathbb{C}^{m\times
n}$ $(m\ge n)$, the flops required for several matrix operations are
summarized in Table \ref{tab:complexity} where the operation
$\otimes$ is defined as $\mathbf{G}\otimes\mathbf{G} =
\mathbf{G}\mathbf{G}^\JHdagger$.

In the MAX-SNR scheme, each user requires one $\otimes$ operation, a
single SVD, $2M$ real additions and $M$ real multiplications to find
feedback information.
Correspondingly, the total computational complexity becomes
$K\times(8N_RM^2-2N_RM)+K\times (24N_R^3 + 48N_R^3 + 54N_R^3) + K
\times 3M = K \times (128N_R^3-N_R^2+\tfrac{3}{2}N_R)$ flops.
In OIA1 scheme, two $\otimes$ operations, two matrix scaling, a
single matrix addition, a single SVD, $2M$ real additions, and $M$
real multiplications are required at each user to find the feedback
information, so the total computational complexity becomes
$K\times 2 (8N_RM^2-2N_RM) + K\times 2N_R^2 + K\times 2N_R^2 +
K\times (24N_R^3+48N_R^3+54N_R^3)  + K \times 3M= K\times (130N_R^3
+3N_R^2 +\tfrac{3}{2}N_R)$ flops.
Note that the postprocessing matrix should be calculated to find
feedback information both in the MAX-SNR and the OIA1 schemes.
On the other hand, the OIA2 scheme requires two Gram-Schmidt
orthogonalization, two $\otimes$ operations, one matrix addition,
and a single $\Vert\cdot\Vert_F$ operation to construct feedback
information. The selected user needs $130N_R^3 + 3N_R^2$ additional
complexity to find the postprocessing matrix. Therefore, the total
complexity of the OIA2 scheme becomes $K \times 4 (8N_RM^2-2N_RM) +
K\times 2N_R^2 + K\times 4N_R^2  + (130N_R^3 + 3N_R^2) = K \times
(8N_R^3+2N_R^2) + (130N_R^3 + 3N_R^2)$.

The computational complexities of various schemes are summarized in
Table \ref{tab:comparison}.
When $N_R=4$, the required computational complexities according to
the number of users are plotted in Fig. \ref{fig:complexity}.
We can observe that the complexity of the OIA2 scheme is about
6.15\% of OIA1 scheme's complexity when the number of receive
antennas, $N_R$, and the number of users, $K$, are sufficiently
large.

%%%%%%%%%%%%%%%%%%%%%%%%%%%%%%%%%%%%%%%%%%%%%%%%%%%%%%%%%%%%%%%%%%%
% % % % % % % % % % % % % % % % % % % % % % % % % % % % % % % % % %
%%%%%%%%%%%%%%%%%%%%%%%%%%%%%%%%%%%%%%%%%%%%%%%%%%%%%%%%%%%%%%%%%%%
\subsection{Performance Comparison}

In Fig. \ref{fig:user50_M1_alpha10_new}, the proposed OIA schemes
are compared with other user selection schemes in terms of
achievable rate per transmitter when $(N_T,M,N_R)=(1,1,2)$ and
$K=50$.
In this case, the optimal user selection scheme is to maximize the
capacity based on \eqref{eqn:optimal_capacity_ik}.
The proposed OIA schemes significantly outperform the MAX-SNR scheme
in the high SNR region.
It is shown that the proposed OIA2 scheme achieves a similar rate to
the OIA1 scheme but it requires much less computational complexity.
It is also shown that the proposed OIA schemes significantly
outperform the conventional MAX-SNR scheme in the high SNR region.
For a finite number of users, the achievable rates using the optimal
scheme, OIA1, and OIA2 schemes are saturated in the high SNR region.
On the other hand, the TDM1 and the TDM2 schemes achieve a DoF of
$1/3$ and $2/3$, respectively, and outperform the OIA schemes above
50dB and 30dB SNR, respectively.

To evaluate practical gains of the proposed OIA schemes, Fig.
\ref{fig:user50_M1_alpha10_new} compares the achievable rate of the
proposed OIA scheme with those of two well-known IA techniques --
Gomadam's MAX-SINR scheme \cite{GCJ2011} and the altering
minimization scheme \cite{PH2009}.
The antenna configuration $(N_T,M,N_R)=(2,1,2)$ is used for both
Gomadam's MAX-SINR scheme and the altering minimization scheme
because DoF of 1 cannot be achieved under the antenna configuration
$(N_T,M,N_R)=(1,1,2)$ used in our system model.
In Gomadam's MAX-SINR scheme \cite{GCJ2011}, the precoding and the
postprocessing matrices are iteratively optimized assuming the
reciprocity of the uplink and downlink channels. In the altering
minimization scheme \cite{PH2009}, perfect CSIT is assumed and
information sharing among the transmitters is allowed. However, it
should be noted that our OIA scheme do not requires perfect CSIT and
transmitter cooperation contrary to \cite{GCJ2011} and
\cite{PH2009}.

%%%%%%%%%%%%%%%%%%%%%%%%%%%%%%%%%%%%%%%%%%%%%%%%%%%%%%%%%%%%%%%%%%%
% % % % % % % % % % % % % % % % % % % % % % % % % % % % % % % % % %
%%%%%%%%%%%%%%%%%%%%%%%%%%%%%%%%%%%%%%%%%%%%%%%%%%%%%%%%%%%%%%%%%%%
\JHnewpage\section{Conclusion}

We have interpreted the opportunistic interference management from a
perspective of IA and proposed a novel opportunistic interference
alignment (OIA) and analyzed its achievable DoF and its user scaling
law in a three-transmitter $M\times 2M$ MIMO IC channel.
The proposed OIA schemes have been shown to achieve $M$ DoF per
transmitter by opportunistically selecting the user whose received
interference signals are most aligned with each other.
Thus, in the high SNR region (i.e., from a DoF aspect), each
transmitter should select the user whose associated interfering
channels are aligned as much as possible.
Contrary to conventional IA which is known to achieve $2M/3$ DoF per
user in a three-transmitter $M \times 2M$ MIMO IC, the proposed OIA
schemes do not sacrifice the spatial dimensions in aligning
interference signals and secure the full spatial DoF by exploiting
the user dimensionality. Furthermore, the proposed OIA schemes do
not require global channel knowledge at the transmitters but need
only scalar value feedback from each user for user selection. We
have also proved that the full DoF of $M$ can be achieved when the
number of users grows with an appropriate scale. Finally, we have
compared our proposed scheme with the conventional schemes. Our
proposed OIA schemes have been shown to have advantages over
conventional user selection schemes for interference mitigation in
terms of both computational complexity and achievable rate.

\newpage
%%%%%%%%%%%%%%%%%%%%%%%%%%%%%%%%%%%%%%%%%%%%%%%%%%%%%%%%%%%%%%%%%%%
% % % % % % % % % % % % % % % % % % % % % % % % % % % % % % % % % %
% % % % % % % % % % % % % % % % % % % % % % % % % % % % % % % % % %
% % % % % % % % % % % % % % % % % % % % % % % % % % % % % % % % % %
% % % % % % % % % % % % % % % % % % % % % % % % % % % % % % % % % %
% % % % % % % % % % % % % % % % % % % % % % % % % % % % % % % % % %
% % % % % % % % % % % % % % % % % % % % % % % % % % % % % % % % % %
%%%%%%%%%%%%%%%%%%%%%%%%%%%%%%%%%%%%%%%%%%%%%%%%%%%%%%%%%%%%%%%%%%%
\linespread{1.8}
%\bibliographystyle{IEEEtran}
%\bibliography{IEEEabrv,[20110115]OIA_references}

\clearpage

\begin{figure}[!b]
\centering
  \includegraphics[width=.5\columnwidth]{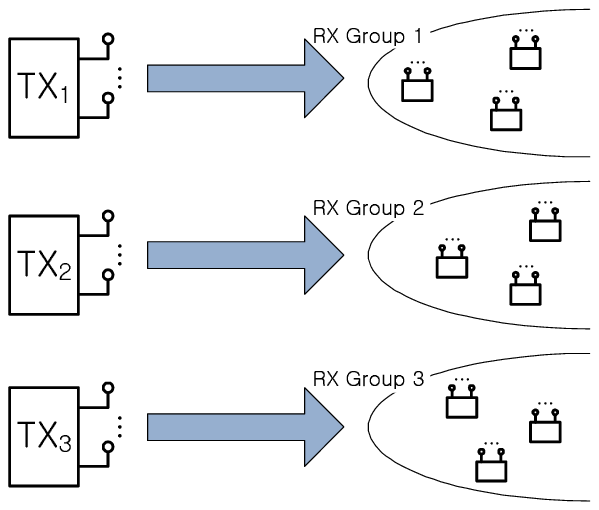}\\
  \caption{System model. Each transmitter selects one user
  from its group. }
  \label{fig:system_model}
\end{figure}

\begin{figure}[!t]
\centering
  \includegraphics[width=.60\columnwidth]{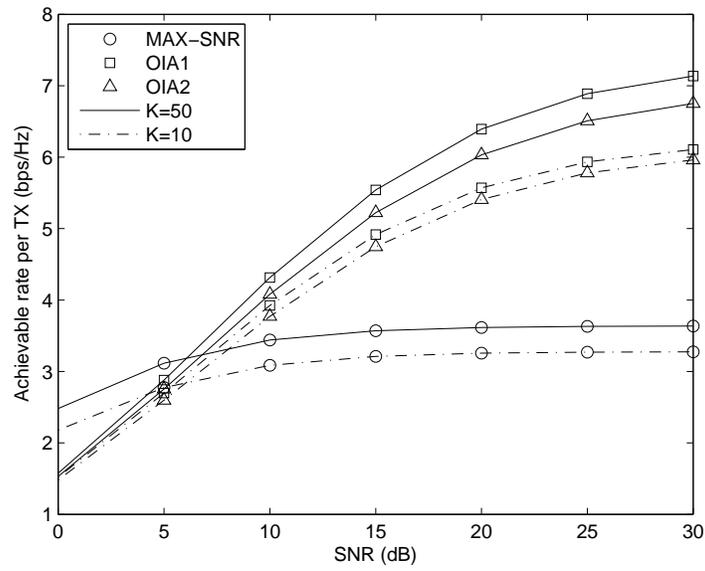}\\
  \caption{The achievable rate per transmitter using various schemes for $K=10, 50$ when $(N_T,M,N_R)=(2,2,4)$.}
  \label{fig:varying_K}
\end{figure}

\clearpage

\begin{figure}[!t]
\centering
  \includegraphics[width=.60\columnwidth]{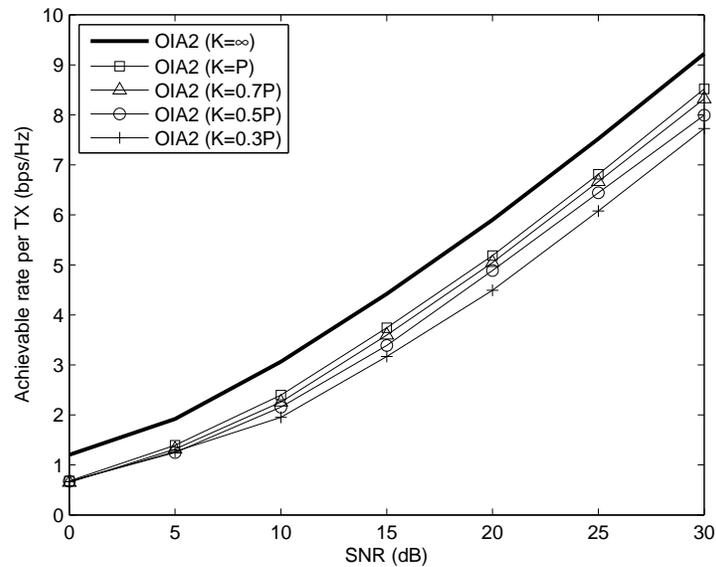}\\
  \caption{The achievable rate per transmitter of OIA2 scheme with
  scaling $K\propto P$ when $(N_T,M,N_R)=(1,1,2)$.}
  \label{fig:user_scaled_alpha10}
\end{figure}

\begin{table}\centering
\caption{The complexity of various operations for
$\mathbf{G}\in\mathbb{C}^{m\times n}$}
\begin{tabular}{c|c}
\hline
Operation & Complexity (flops) \\
\hline
\hline $\alpha\mathbf{G}$, $\mathbf{G}+\mathbf{G}$ & $2mn$\\
\hline $\Vert\mathbf{G}\Vert_F$ & $4mn$\\
\hline $\mathbf{G}\otimes\mathbf{G}=\mathbf{G}\mathbf{G}^\JHdagger$ & $8mn^2-2mn$\\
\hline Gram-Schmidt Ortho. & $8mn^2-2mn$\\
\hline Singular Value Decomp. & $24m^2n+48mn^2+54n^3$\\
\hline
\end{tabular}
\label{tab:complexity}
\end{table}

\clearpage

\begin{table}\centering
\caption{The complexity of various schemes}
\begin{tabular}{c|c|c}
    \hline
    Scheme &  Complexity & $\underset{(K, N_R\to \infty)}{\textrm{Ratio}}$\\
    \hline \hline
    MAX-SNR &  $K \times (128N_R^3-N_R^2+\tfrac{3}{2}N_R)$ & 98.4\%\\
    \hline
    OIA1 &  $K\times (130N_R^3 +3N_R^2 +\tfrac{3}{2}N_R)$ & 100 \%\\
    \hline
    OIA2   &  $K \times (8N_R^3+2N_R^2) + (130N_R^3 + 3N_R^2)$ & 6.15\%\\
    \hline
\end{tabular}
\label{tab:comparison}
\end{table}

\begin{figure}[!b]
\centering
  \includegraphics[width=.60\columnwidth]{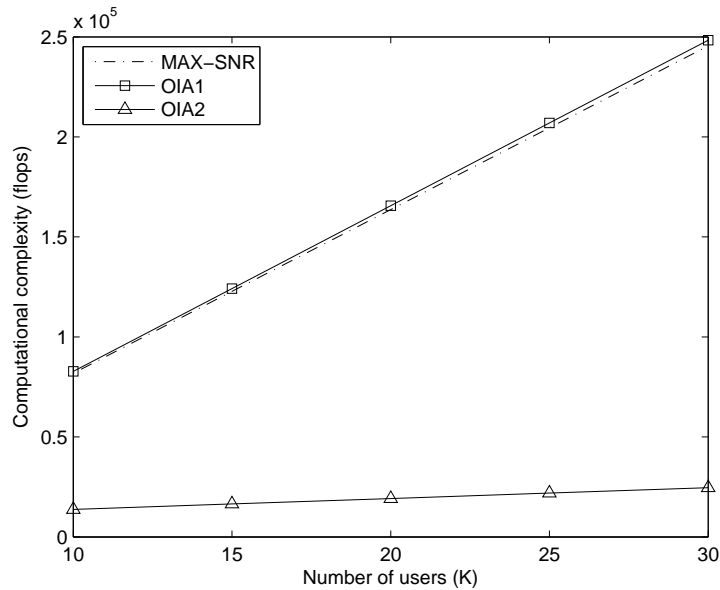}\\
  \caption{Complexities of various user selection schemes according to the number of users $K$ when $N_R=4$.}
  \label{fig:complexity}
\end{figure}

\begin{figure}[!t]
\centering
  \includegraphics[width=.60\columnwidth]{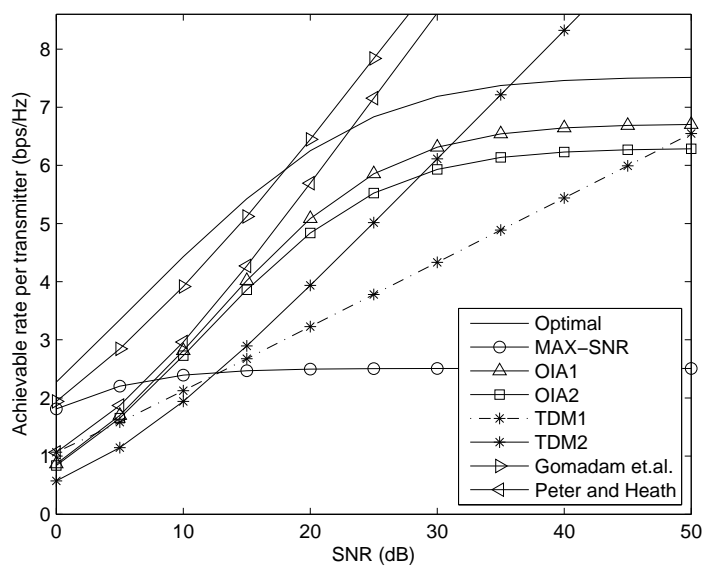}\\
  \caption{The achievable rate per transmitter of various user selection
  schemes when $(N_T,M,N_R)=(1,1,2)$ and $K=50$.}
  \label{fig:user50_M1_alpha10_new}
\end{figure}

\end{document}